\documentclass[11pt]{article}
\usepackage[T1]{fontenc}
\usepackage[utf8]{inputenc}
\usepackage{microtype}
\usepackage{graphicx}
\usepackage{fullpage}
\usepackage{amsmath}
\usepackage{amsthm}
\usepackage{amsfonts}
\usepackage{xcolor}
\usepackage[numbers]{natbib}
\usepackage[ruled,vlined,linesnumbered]{algorithm2e}
\usepackage{aliascnt}

\usepackage{hyperref}
\usepackage{cleveref}

\newtheorem{theorem}{Theorem}[section]

\newaliascnt{remark}{theorem}

\aliascntresetthe{remark}

\newaliascnt{definition}{theorem}
\newtheorem{definition}[definition]{Definition}
\aliascntresetthe{definition}

\newaliascnt{lemma}{theorem}
\newtheorem{lemma}[lemma]{Lemma}
\aliascntresetthe{lemma}

\newaliascnt{claim}{theorem}
\newtheorem{claim}[claim]{Claim}
\aliascntresetthe{claim}

\newaliascnt{corollary}{theorem}
\newtheorem{corollary}[corollary]{Corollary}
\aliascntresetthe{corollary}

\crefname{theorem}{Theorem}{Theorems}
\Crefname{theorem}{Theorem}{Theorems}
\crefname{remark}{Remark}{Remarks}
\Crefname{remark}{Remark}{Remarks}
\crefname{definition}{Definition}{Definitions}
\Crefname{definition}{Definition}{Definitions}
\crefname{lemma}{Lemma}{Lemmas}
\Crefname{lemma}{Lemma}{Lemmas}
\crefname{claim}{Claim}{Claims}
\Crefname{claim}{Claim}{Claims}
\crefname{corollary}{Corollary}{Corollaries}
\Crefname{corollary}{Corollary}{Corollaries}

\newcommand{\cF}{\mathcal F}
\newcommand{\cS}{\mathcal S}

\newcommand{\cG}{\mathcal G}

\newcommand{\cC}{\mathcal C}
\newcommand{\cA}{\mathcal A}
\newcommand{\cB}{\mathcal B}
\newcommand{\cR}{\mathcal R}
\newcommand{\cX}{\mathcal X}
\newcommand{\gendistfamily}{\mathcal M}
\newcommand{\R}{\mathbb R}
\newcommand{\E}{\mathbb E}
\newcommand{\Z}{\mathbb Z}
\newcommand{\I}{\mathbb I}

\newcommand{\ind}{\mathsf{Ind}}
\newcommand{\kl}{\mathsf{KL}}
\newcommand{\softmax}{\mathsf{softmax}}

\newcommand{\poly}{\mathsf{poly}}
\newcommand{\polylog}{\mathsf{polylog}}

\newcommand{\be}{\mathbf e}

\newcommand{\CP}{\mathsf{CP}}

\newcommand{\Havgmin}{\mathsf{H_{\textup{\textsf{avg-min}}}}}
\newcommand{\HSh}{\mathsf{H_{\textup{\textsf{Sh}}}}}
\renewcommand{\H}{\mathsf{H}}
\newcommand{\D}{\mathsf{D}}
\newcommand{\Dkl}{\mathsf{D}_{\kl}}
\newcommand{\Supp}{\mathsf{Supp}}
\newcommand{\Time}{\mathsf{Time}}

\newcommand{\eps}{\varepsilon}
\newcommand{\sps}[1]{^{(#1)}}

\DeclareMathOperator*{\argmax}{arg\,max}

\title{Generalized and Unified Equivalences between\\ Hardness and Pseudoentropy}
\author{Lunjia Hu\thanks{Supported by the Simons Foundation Collaboration on the
Theory of Algorithmic Fairness and the Harvard Center for Research on Computation and Society. Part of this work was performed while LH was a Postdoctoral Fellow at Harvard University.} \\ Northeastern University \\ \texttt{lunjia@alumni.stanford.edu} \and Salil Vadhan\thanks{Supported by a Simons Investigators Award.} \\ Harvard University \\ \texttt{salil\_vadhan@harvard.edu}}
\date{}

\allowdisplaybreaks
\begin{document}

\maketitle
\begin{abstract}
Pseudoentropy characterizations give quantitatively precise formulations of the relationship between computational \emph{hardness} and computational \emph{randomness}.
We prove a unified pseudoentropy characterization that generalizes and strengthens previous results in both uniform and nonuniform models of computation.
Our characterization applies to a general family of entropy notions, including Shannon entropy and min-entropy as special cases.
Moreover, the characterizations for these different entropy notions can be witnessed simultaneously by a single universal function, which captures both computational hardness and computational randomness.

A key technical insight is that \emph{weight-restricted calibration}, from the recent literature on algorithmic fairness, together with standard computational indistinguishability (known as \emph{multiaccuracy} in the fairness literature), suffices for proving pseudoentropy characterizations for general entropy notions.
To obtain this combination of properties, we
prove an enhanced version of the Leakage Simulation Lemma (Jetchev and Pietrzak, 2014), which in turn extends the Complexity Theoretic-Regularity Lemma (Trevisan, Tulsiani, and Vadhan, 2009) from boolean functions to ones over a larger alphabet.
Our Enhanced Regularity/Leakage-Simulation Lemma enables us to obtain an exponential improvement in the dependence on the alphabet size compared with the pseudoentropy characterizations of Casacuberta, Dwork, and Vadhan (2024), which are based on the stronger notion of \emph{multicalibration}.
We also show that this exponential dependence on the alphabet size is inevitable for multicalibration and even for the weaker notion of calibrated multiaccuracy.

Our Enhanced Regularity/Leakage-Simulation Lemma can be viewed as an extension of a result about ``sample-access outcome indistinguishability'' (Dwork, Kim, Reingold, Rothblum, and Yona, 2021) from boolean functions to ones over a larger alphabet.
We hypothesize that these lemmas
will prove to be a powerful tool for other applications in average-case complexity and cryptography.
\end{abstract}

\newcommand{\zo}{\{0,1\}}
\newcommand{\negl}{\mathrm{negl}}

\section{Introduction}

\subsection{Background on Hardness--Randomness Equivalences}
\label{sec:background}

The close relationship between computational {\em hardness} and computational {\em randomness} is central in cryptography and complexity theory.  A classic example of this relationship is Yao's equivalence between pseudorandomness and (maximal) unpredictability~\cite{Yao82}, one form of which is the following:

\begin{theorem}
Let $(X,Y)$ be a random variable distributed on $\zo^n\times \zo^\ell$ with $\ell=O(\log n)$.  Then the following are equivalent:
\begin{enumerate}
\item $(X,Y)\approx^c (X,U_\ell)$, where $\approx^c$ denotes computational indistinguishability (against nonuniform polynomial-time algorithms).
\item \label{cond:max-unpred}
For every nonuniform polynomial-time algorithm $A$,
$$\Pr[A(X)= Y] \leq \frac{1}{2^\ell}+\negl(n).$$
\end{enumerate}
\end{theorem}

We refer to \Cref{cond:max-unpred} as {\em maximal} unpredictability, because it is trivial for an efficient $A$ to achieve prediction probability $1/2^\ell$, by just outputting a uniformly random $\ell$-bit string.
It is natural to ask what happens if we weaken the hardness of prediction to allow $A$ to succeed with some probability between $1/2^\ell$ and 1. Vadhan and Zheng~\cite{characterize,zheng} showed that such weak unpredictability is equivalent to $Y$ having high {\em pseudo-average-min-entropy}~\cite{hlr} given $X$.
\begin{theorem}
\label{thm:intro-min-rv}
Let $(X,Y)$ be a random variable distributed on $\zo^n\times \zo^\ell$ with $\ell=O(\log n)$, and let $k\in [0, \ell]$.  Then the following are equivalent:
\begin{enumerate}
\item \label{itm:intro-min-rv-pseudo}
There is a random variable $Z$ on $\zo^\ell$, jointly distributed with $X$, such that $(X,Y)\approx^c (X,Z)$ and
$\Havgmin(Z|X)\geq k$, where $\Havgmin$ denotes {\em average min-entropy}~\cite{DodisOsReSm}:
\[
\Havgmin(Z|X) := \log_2\left(\frac{1}{\E_{x\sim X}\left[\max_{z_x\in \zo^\ell}\Pr[Z = z_x|X=x]\right].}\right)
\]
\item \label{itm:intro-min-rv-unpred}
For every nonuniform polynomial-time algorithm $A$,
$$\Pr[A(X)= Y] \leq \frac{1}{2^k}+\negl(n),$$ where
$\negl(n)$ denotes a negligible function.
\end{enumerate}
\end{theorem}
Note that $\Havgmin(Z|X)\geq k$ is equivalent to
saying that for {\em every} function $A$ (even of unbounded complexity), $\Pr[A(X)=Z]\leq 2^{-k}$.  (Indeed, the best strategy is to set $A(x)$ to be $z_x$ that maximizes 
$\Pr[Z = z_x|X=x]$ for each $x$, which will achieve success probability exactly the expectation that appears in the definition of $\Havgmin(Z|X)$.) Thus, \Cref{thm:intro-min-rv} shows that computational unpredictability (as captured by \Cref{itm:intro-min-rv-unpred}) is equivalent to indistinguishability from information-theoretic unpredictability (as captured by \Cref{itm:intro-min-rv-pseudo}).
As discussed in \cite{characterize,zheng}, the case when $Y$ consists of $\ell=1$ bits is equivalent to Holenstein's optimal version of Impagliazzo's Hardcore Lemma~\cite{Impagliazzo95,Holenstein05}; the novelty in their result was proving it for all $\ell=O(\log n)$.

Now it is also natural to ask what happens if we replace the average min-entropy above with other information-theoretic measures of randomness, like conditional Shannon entropy.  In this case, Vadhan and Zheng~\cite{characterize} gave the following characterization.
\begin{theorem}
\label{thm:intro-shannon-rv}
Let $(X,Y)$ be a random variable distributed on $\zo^n\times \zo^\ell$ with $\ell=O(\log n)$, and let $k\geq 0$.  Then the following are equivalent:
\begin{enumerate}
\item There is a random variable $Z$ on $\zo^\ell$, jointly distributed with $X$, such that $(X,Y)\approx^c (X,Z)$ and
$\HSh(Z|X)\geq \HSh(Y|X)+k - \negl(n)$, where $\HSh$ denotes the conditional Shannon entropy.
\item For every nonuniform, randomized polynomial-time algorithm $A$, $$\Dkl((X,Y)\|(X,A(X))) \geq k-\negl(n),$$ where $\Dkl$ denotes KL divergence (a.k.a. relative entropy).
\end{enumerate}
\end{theorem}
Here the hardness of $A(X)$ trying to predict the value of $Y$ is replaced by $A(X)$ trying to get close (in KL divergence) to the conditional distribution of $Y$ given $X$.  Vadhan and Zheng~\cite{characterize} used \Cref{thm:intro-shannon-rv}
to give a simpler construction of pseudorandom generators from one-way functions. Their pseudorandom generator construction avoids the use of hardcore bits like Goldreich--Levin~\cite{goldreich-levin}, which was the traditional way of converting the hardness of inverting a one-way function into pseudoentropy, as was used in the original construction of H\aa stad, Impagliazzo, Levin, and Luby~\cite{one-way}.

Note that \Cref{thm:intro-shannon-rv} is concerned with the {\em pseudoentropy gap} $k=\HSh(Z|X)-\HSh(Y|X)$, which measures how much \emph{more} randomness $Y$ appears to have to feasible observers in comparison to computationally unbounded ones.  Correspondingly, even if $\HSh(Y|X)$ is large, $Y$ may still be considered ``easy'' in that an efficient $A(X)$ can sample accurately from the conditional distribution $Y|_X$.  In contrast, the parameter $k$ in \Cref{thm:intro-min-rv} does not distinguish between the computational hardness and information-theoretic hardness in $Y$ given $X$.

Aside from this difference, the two theorems are very similar in spirit and one may wonder whether there is a unified proof for them, one which ideally also generalizes to other measures of randomness (beyond min-entropy and Shannon entropy).  Zheng~\cite{zheng} made progress in this direction, considering a wide family of entropy measures $\H_\varphi$ and their corresponding {\em Bregman divergences} $\D_\varphi$ (whose definition we defer to later in the introduction).

\begin{theorem} \label{thm:intro-bregman-rv}
Suppose that $\ell=O(\log n)$ and let $\H_\varphi(B|A)$ be a function that takes a random variable $(A,B)$ distributed on $\zo^n\times \zo^\ell$ and outputs the negative average over $a\gets A$ of a differentiable and strictly convex function $\varphi$ applied to the probability mass function of $B|A=a$.  Let $\D_\varphi$ be the Bregman divergence associated with $\varphi$.  Suppose that given a linear function $\Lambda$, the probability mass function $p$ on $\zo^\ell$ that minimizes $\Lambda(p)+\varphi(p)$ can be computed by nonuniform circuits of size $\poly(n)$, and that, given $p$, the gradient $\nabla \varphi(p)$ can also be computed accurately by circuits of size $\poly(n)$. Also assume that $\sup_p\|\nabla \varphi(p)\|_\infty = \poly(n)$.

Let $(X,Y)$ be a random variable distributed on $\zo^n\times \zo^\ell$ with $\ell=O(\log n)$, and let $k\geq 0$.  Then the following are equivalent:
\begin{enumerate}
\item \label{item:bregman-1} There is a random variable $Z$ on $\zo^\ell$, jointly distributed with $X$, such that $(X,Y)\approx^c (X,Z)$ and
$\H_\varphi(Z|X)\geq \H_\varphi(Y|X)+k - \negl(n)$.
\item \label{item:bregman-2} For every nonuniform, randomized polynomial-time algorithm $A$, $$\D_\varphi((X,Y)\|(X,A(X))) \geq k-\negl(n).$$
\end{enumerate}
\end{theorem}

This theorem implies \Cref{thm:intro-shannon-rv} as a special case, because if we take $\varphi(p)=\sum_b p_b\log p_b$, then $\H_\varphi$ becomes conditional Shannon entropy and $\D_\varphi$ becomes KL divergence.\footnote{In this case, $\|\nabla \varphi(p)\|_\infty$ is unbounded as $p$ approaches the boundary of $\Delta_L$, violating the boundedness assumption required by \Cref{thm:intro-bregman-rv}, but this can be handled by slightly perturbing $p$ away from the boundary (see \Cref{sec:unbounded}).}  Moreover, if we instead take $\varphi(p)= \|p\|_\infty :=  \max_b p_b$, then $\H_\varphi$ is equivalent (up to a monotone increasing change of variables) to average min-entropy.
Unfortunately, we cannot derive \Cref{thm:intro-min-rv} from \Cref{thm:intro-bregman-rv} because this choice of $\varphi$ is not differentiable, and the proof of \Cref{thm:intro-bregman-rv} crucially relies on differentiability.  However, it does at least give us generalizations to other entropy measures.  In particular, taking $\varphi(p) = \sum_b p_b^2$ or $\varphi(p) = \sqrt{\sum_b p_b^2}$,
$\H_\varphi(B|A)$ measures conditional forms of the collision probability (equivalently, R\'enyi entropy of order 2), which have proved useful in cryptography and algorithmic applications, for example in \cite{HaitnerNgOnReVa,ChungMiVa}.  (For interpretations of $\D_\varphi$ in these cases, see \Cref{sec:examples}.)

A different unification of the hardness--randomness equivalences came recently in the work of Casacuberta, Dwork, and Vadhan~\cite{cdv}. Inspired by the concurrent work of Dwork, Lee, Lin, and Tankala\ \cite{dllt}, they reinterpreted the Multicalibration Theorem from the algorithmic fairness literature~\cite{hkrr,omni} and showed how it gives a simpler proof of \Cref{thm:intro-min-rv,thm:intro-shannon-rv} when $\ell=O(1)$, along with other known results about average-case complexity and computational indistinguishability (namely, Impagliazzo's Hardcore Lemma~\cite{Impagliazzo95,Holenstein05} and the complexity-theoretic Dense Model Theorem~\cite{rttv}).  Their interpretation of the Multicalibration Theorem can be stated informally as follows:
\begin{theorem}[Multicalibration Theorem, informally stated]
\label{thm:intro-mc-rv-informal}
Let $(X,Y)$ be a random variable distributed on $\zo^n\times \zo^\ell$ with $\ell=O(\log n)$.  Then there is a function $P : \zo^n\rightarrow \zo^{O(2^\ell\cdot \log n)}$ with circuit complexity $\poly(n)^{2^\ell}$ such that $(X,Y)\approx^c (X,Z)$, where for every $x\in \Supp(X)$, we define $Z|_{X=x}$ to be identically distributed to $Y|_{P(X)=P(x)}$.
\end{theorem}
Notice that conditioned on the value of $P(X)$, $Z$ and $X$ are independent of each other.  This means that, conditioned on $P(X)$, $(X,Z)$ is the ``most random'' random variable of the form $(X,W)$ that is indistinguishable from $(X,Y)$.  (Since $\ell=O(\log n)$, if $(X,W)$ is computationally indistinguishable from $(X,Y)$, then the marginal distribution of $W$ must be statistically indistinguishable from $Y$, so the only way to noticeably increase the entropy of $W$ given $X$ is to make $W$ more independent of $X$.)

Furthermore, when $\ell=O(1)$, $P$ has polynomial circuit complexity and outputs only $O(\log n)$ bits.  In this setting, we can thus derive \Cref{thm:intro-min-rv,thm:intro-shannon-rv,thm:intro-bregman-rv} easily.  To illustrate with \Cref{thm:intro-min-rv}, consider any polynomial-size nonuniform predictor $A$.  Since $(X,Y)\approx^c (X,Z)$, we have:
$$\Pr[A(X)=Y] \leq \Pr[A(X)=Z]+\negl(n) \leq 2^{-\Havgmin(Z|X)}+\negl(n),$$
and furthermore, we can construct an $A$ with $\Pr[A(X)=Y] = 2^{-\Havgmin(Z|X)}$, by having $A(X)$ compute $q=P(X)$ and output $b_q=\argmax_b \Pr[Y=b|P(X)=q]$, which we can hardwire into a polynomial-size nonuniform predictor.  We can afford to hardwire all of the $b_q$ values since $P$ outputs $O(\log n)$ bits and hence there are only $2^{O(\log n)}=\poly(n)$ choices for $q$.
Similarly, the fact that $P$ has polynomial circuit complexity and only outputs $O(\log n)$ bits implies that $Z$ is the random variable of maximum average min-entropy given $X$ such that $(X,Y)\approx^c (X,Z)$.

Interestingly, for \Cref{thm:intro-shannon-rv,thm:intro-bregman-rv}, we get that the polynomial-size nonuniform predictor $A$ that minimizes the divergence $\D_\varphi((X,Y)\|(X,A(X)))$ simply samples from $Z$ given $X$ (which can be done efficiently since $P$ has polynomial circuit complexity and outputs only $O(\log n)$ bits).
That is, the best efficiently computable predictor for $Y$ given $X$ (with respect to $\D_\varphi$) is {\em identical} to the highest-$\H_\varphi$ distribution $Z$ that is indistinguishable from $Y$ given $X$.  Furthermore, this distribution $Z$ is the same for all of the Bregman entropy measures $\H_\varphi$
(as well as for $\Havgmin$).

\paragraph{Correcting the order of quantifiers.}  The informal statement above suppresses an important order-of-quantifiers issue in the Multicalibration Theorem.  The function $P$ does not have circuit complexity $\poly(n)^{2^\ell}$ for a fixed polynomial $\poly(\cdot)$.  Rather, for every polynomial $p(n)$, there is a $P$ of circuit complexity $\poly(p(n))^{2^\ell}$, outputting $\log\left(\poly(p(n))^{2^\ell}\right)$ bits, such that $Y$ and $X$ are computationally indistinguishable from being independent conditioned on $P(X)$, for nonuniform distinguishers of size $p(n)$ and with an advantage bounded by $1/p(n)$.  This corrected order of quantifiers still suffices to prove \Cref{thm:intro-min-rv,thm:intro-shannon-rv,thm:intro-bregman-rv} when $\ell=O(1)$, without modifying those theorem statements.  However, the order of quantifiers must also be incorporated into our informal statement that the best efficiently computable predictor for $Y$ given $X$ (with respect to any $\D_\varphi$) is {\em identical} to the highest-entropy distribution $Z$ that is indistinguishable from $Y$ given $X$.  We give precise statements below.

In summary, the Multicalibration Theorem comes close to providing a unified picture of hardness versus randomness, but it has two significant deficiencies.  First, it has a factor of $2^\ell$ in the \emph{exponent} of the circuit complexity of $P$, which is unfortunately necessary (as we prove in \Cref{cor:lb}).  Second, it is only stated and used in \cite{cdv} for nonuniform circuit complexity, whereas the prior proofs of \Cref{thm:intro-min-rv,thm:intro-shannon-rv} had uniform-complexity analogues~\cite{unif-minmax,zheng}.

\subsection{Our Contributions}
\label{sec:contributions}
In this work, we address the aforementioned deficiencies of both \Cref{thm:intro-bregman-rv,thm:intro-mc-rv-informal}, and thereby provide a more unified and clarified understanding of hardness versus randomness.  Specifically, we provide a substantial strengthening of \Cref{thm:intro-bregman-rv} with the following features in \Cref{thm:intro-main-1}:
\begin{itemize}
\item The convex function $\varphi$ does not need to be differentiable, and thus our generalization encompasses \Cref{thm:intro-min-rv}.
\item Our theorem applies simultaneously to a set $\Phi$ of functions $\varphi$, with a complexity blow-up that depends on the circuit complexity needed to accurately evaluate the (sub)gradients $\nabla \varphi$, for an arbitrary $\varphi\in \Phi$.  For the specific $\varphi$ needed to capture \Cref{thm:intro-min-rv,thm:intro-shannon-rv}, as well as for conditional collision probability (R\'enyi 2-entropy), the complexity blow-up is just $\poly(2^\ell)$, which is affordable for $\ell=O(\log n)$.  But if we take all possible functions $\varphi$, we get a blow-up that is doubly exponential in $\ell$, as in the Multicalibration Theorem.
\item Similarly to the Multicalibration Theorem, we deduce (subject to appropriate quantifiers as mentioned above) that the best efficiently computable prediction for $Y$ given $X$ (with respect to $\D_\varphi$) is {\em identical} to the highest-$\H_\varphi$ distribution $Z$ that is indistinguishable from $Y$ given $X$.  Furthermore, this distribution $Z$ is the same for all of the $\varphi\in \Phi$.
\item We state and prove a uniform-complexity generalization of our theorem in \Cref{thm:intro-unif}.
This result encompasses the uniform-complexity versions of \Cref{thm:intro-min-rv,thm:intro-shannon-rv} as shown in \cite{characterize,zheng} and extends \Cref{thm:intro-bregman-rv} to the uniform-complexity setting as well (which, to our knowledge, has not been achieved in prior work). The uniform-complexity setting plays an important role in \cite{characterize}, which uses the uniform-complexity version of \Cref{thm:intro-shannon-rv} to give refined constructions of pseudorandom generators from one-way functions.
\item To justify the blow-up in our theorem depending on the complexity of (approximately) computing the (sub)gradients $\nabla \varphi$, we prove that a doubly exponential dependence on $\ell$ is unavoidable if we do not place any computational efficiency assumptions on $\varphi$, even when $\Phi$ only contains a single convex function $\varphi$ with bounded subgradients (\Cref{thm:lb-exp}). This result implies that the doubly exponential dependence is necessary for achieving multicalibration (as well as the weaker notion of \emph{calibrated multiaccuracy}) even when there is only a single, extremely simple distinguisher (\Cref{cor:lb}).

\item A key insight in our work is that full multicalibration is stronger than necessary for proving \Cref{thm:intro-min-rv,thm:intro-shannon-rv}.  Instead, it suffices to use a combination of two substantially weaker notions, {\em multiaccuracy} and {\em weight-restricted calibration}.  The latter notion is tailored to the specific family $\Phi$, imposing one constraint per function $\varphi\in \Phi$, leading to a complexity blow-up that depends on the maximum complexity of (the gradients of) the functions in $\Phi$.
In fact, we prove an {\em Enhanced Regularity/Leakage-Simulation Lemma}, which smoothly interpolates between these notions and multicalibration, as discussed in the following section.  
\end{itemize}

In independent works, Casacuberta, Gopalan, Kanade, and Reingold~\cite{global-cal-ma} and Dwork and Tankala~\cite{Tankala25} demonstrated the power of (variants of) a similar notion termed \emph{calibrated multiaccuracy} (see \Cref{sec:proof-idea} for more discussion). Casacuberta et al.~\cite{global-cal-ma} give applications in agnostic learning and to proving Impagliazzo's Hardcore Lemma (which, as mentioned earlier, is equivalent to \Cref{thm:intro-mc-rv-informal} for the case $\ell=1$), while Dwork and Tankala~\cite{Tankala25} show that calibrated multiaccuracy can substitute for multicalibration in the results of \cite{MarcussenPuVa25}, which characterize computational indistinguishability under repeated samples.

\subsection{Enhanced Regularity/Leakage-Simulation Lemma}

As mentioned above, our results are based on the following variant of the Multicalibration Theorem (where we suppress the same important order-of-quantifiers as we did in \Cref{thm:intro-mc-rv-informal}).

\begin{theorem}[Enhanced Regularity/Leakage-Simulation Lemma, informally stated] \label{thm:intro-enhanced-informal}
Let $(X,Y)$ be a random variable distributed on $\zo^n\times \zo^\ell$ with $\ell=O(\log n)$. Then there is a randomized function $S : \zo^n\rightarrow \zo^\ell$ with circuit complexity $\poly(n,2^\ell)$ such that $$(X,Y,s(X)) \approx^c (X,S(X),s(X)),$$
where $s(x)\in [0,1]^{2^\ell}$ denotes the probability mass function of $S(x)$, which is computable by a deterministic circuit also of size $\poly(n,2^\ell)$.
\end{theorem}
The naming of this lemma comes from the fact that it extends Complexity-Theoretic Regularity Lemma of Trevisan, Tulsiani, and Vadhan~\cite{ttv} ($\ell=1$) and the Leakage Simulation Lemma of Jetchev and Pietrzak~\cite{leakage} ($\ell>1$), which give the weaker conclusion 
that $$(X,Y) \approx^c (X,S(X)).$$
The terminology ``Regularity Lemma'' comes from he fact that it says that every randomized function $g^* : \zo^n\rightarrow \zo^\ell$, no matter how hard to compute, is indistinguishable from a ``low-complexity'' function $S$.  (Take $Y=g^*(X)$.)  This is analogous to regularity theorems in combinatorics, e.g. the Sz\'emeredi~\cite{szemeredi} and Frieze--Kannan~\cite{frieze-kannan} Regularity Lemmas, which say that every graph is ``indistinguishable'' from a ``low-complexity'' graph (e.g. one that can be described by a number of parameters that depends only on the level of indistinguishability).  The connection between complexity-theoretic regularity as in \cite{ttv} and regularity in graph theory and combinatorics can be made precise; see \cite{ttv,skorski,dllt}.
The terminology ``Leakage-Simulation Lemma'' comes from cryptographic applications where the function $g^*$ represents side information that an adversary gets, for example by a side channel that leaks information on a secret correlated with the $X$.  In this setting, it is natural and important to allow for multiple bits of leakage ($\ell>1$).

The new realization represented in \Cref{thm:intro-enhanced-informal}, coming from the algorithmic fairness literature, is that it is also possible to fool distinguishers that are given the probability distribution of $S(x)$, and this is very useful in applications.  

The $\ell=1$ case of \Cref{thm:intro-enhanced-informal} was proven under the formulation of ``Sample-Access Outcome Indistinguishability'' by Dwork, Kim, Reingold, Rothblum, and Yona~\cite{oi}, who observed that it is essentially equivalent to Multicalbration.  That observation generalizes to any $\ell=O(1)$: if we round each probability mass in the output distribution $S(x)$ to a multiple of $\eps/2^\ell$, the distinguishing advantage will increase by at most $\eps$.  If $\eps=1/\poly(n)$, then the resulting discretized function $s(x)$ defines an efficiently computable partition of $\zo^n$ into at most $(1/\eps)^{2^\ell} = \poly(n)^{2^\ell}$ pieces. When $\ell=O(1)$, the number of pieces is $\poly(n)$, so $(X,Y)$ and $(X,S(X))$ must be indistinguishable even conditioned on the value of $s(X)$, since a $\poly(n)$-sized distinguisher can incorporate optimal $\poly(n)$-sized distinguishers on each piece.  Note that conditioned on $s(X)=v$, $S(X)$ is independent from $X$, just like \Cref{thm:intro-mc-rv-informal} ensures that $Z$ is independent from $X$ conditioned on the value of $P(X)$.  Moreover, we can then modify $S(X)$ to equal the conditional distribution of $Y$ given $s(X)$, exactly matching the statement of the Multicalibration Theorem.

Thus our Enhanced Regularity/Leakage-Simulation Lemma smoothly interpolates between Multicalibration (with its necessarily exponenital dependence on the size $L=2^\ell$ of the label space) and a natural strenghtening of the guarantees of the earlier Regularity/Leakage-simulation lemmas (while keeping their polynomial dependence on $L$). 
Thus, when seeking applications, a fruitful strategy can be to first use the stronger and more intuitive property of Multicalibration, and afterwards see if it can be replaced with the Enhanced Regularity/Leakage-Simulation Lemma to obtain improved parameters. 

Before moving on to our proof ideas (which requires introducing some new formalism), we provide a precise statement of the Enhanced Regularity/Leakage-Simulation Lemma (with no cheating in the order of quantifiers): 

\begin{theorem}[Enhanced Regularity/Leakage-Simulation Lemma for General Distinguishers]
\label{thm:intro-reg-non-unif-rv}
For every $n, \ell, T\in \mathbb{N}$,  $\varepsilon\in (0,1/2)$, and pair of random variables $(X,Y)$ taking values in $\zo^n\times \zo^\ell$, there exists a randomized simulator $S :\zo^n\to \zo^\ell$ such that
    \begin{enumerate}
        \item $S$ has circuit complexity
        \[
        O\left(\frac{T\cdot \ell}{\varepsilon^2} + 2^\ell\cdot \poly(1/\varepsilon,\ell)\right);
        \]
        \item For every distinguisher $D$ of size at most $T$, we have
        $$\left|\Pr[D(X,Y,s(X))=1] - \Pr[D(X,S(X),s(X))=1]\right|\leq \eps,$$
        where for $x\in \zo^n$, $s(x)\in [0,1]^{2^\ell}$ denotes the probability mass function of $S(x)$,
        which is computable by a deterministic circuit of the same size as promised for $S$.
    \end{enumerate}
\end{theorem}
Note that when the distinguisher size $T$, inverse distinguishing advantage $1/\eps$, and range size $2^\ell$ are bounded by $\poly(n)$, the circuit complexity of the simulator $S$ is bounded by $\poly(n)$, albeit a larger polynomial.  (The same increase in circuit size occurs in the Complexity-Theoretic Regularity Lemma and the Leakage Simulation Lemma, and does not interfere with their applicability in complexity theory and cryptography.)

For completeness, we state a similarly precise version of our formulation of the Multicalibration Theorem (\Cref{thm:intro-mc-rv-informal}), as follows from \Cref{thm:intro-reg-non-unif-rv}, in case it is useful in other applications.
\begin{theorem}[Multicalibration Theorem, informally stated]
\label{thm:intro-mc-rv-informal-ours}
For every $n, \ell, T\in \mathbb{N}$,  $\varepsilon\in (0,1/2)$, and pair of random variables $(X,Y)$ taking values in $\zo^n\times \zo^\ell$, there exists a function $P :\zo^n\to [k]$ such that
\begin{enumerate}
    \item $P$ partitions the domain $\zo^n$ into 
    \[
    k=O\left(\frac{1}{\eps}\right)^{2^\ell-1}
    \]
    pieces.
    \item $P$ has circuit complexity
    \[
    O\left(\frac{T\cdot \ell\cdot 2^\ell}{\varepsilon^2}\right)
    + O\left(\frac{1}{\varepsilon}\right)^{2^\ell+1}
    + 2^\ell\cdot \poly(1/\varepsilon,\ell).
    \]
    \item For every distinguisher $D$ of size at most $T$, we have
        \[
        \sum_{i=1}^k
        \left|\Pr[D(X,Y)=1, P(X)=i] - \Pr[D(X,Z)=1, P(X)=i]\right|
        \le \eps,
        \]
        where for $x\in \zo^n$, the $Z|X=x$ is defined to be identically distributed to 
        $Y|P(X)=P(x)$.    
\end{enumerate}
\end{theorem}
Other precise formulations of the Multicalibration Theorem can be found in \Cref{sec:mc}.

The functional formulation and proof of this statement appear in \Cref{thm:intro-mc-rv-formal}.

\subsection{Example Application}

To illustrate the power of the Enhanced Regularity/Leakage-Simulation Lemma, we sketch here how it can be used to give a short proof of the difficult direction of \Cref{thm:intro-min-rv}, namely that unpredictability (\Cref{itm:intro-min-rv-unpred}) implies pseudoentropy (\Cref{itm:intro-min-rv-pseudo}).  For readability, we will use the informally stated \Cref{thm:intro-enhanced-informal}, but the proof using the precise version follows the same argument.

Let $(X,Y)$ be a random variable distributed on $\zo^n\times \zo^\ell$ with $\ell=O(\log n)$.  Per \Cref{itm:intro-min-rv-unpred} of \Cref{thm:intro-min-rv}, assume that for every nonuniform polynomial-time algorithm $A$,
\begin{equation} \label{eqn:intro-min-rv-unpred}
    \Pr[A(X)= Y] \leq \frac{1}{2^k}+\negl(n),
\end{equation}
where
$\negl(n)$ denotes a negligible function.  Now apply the Enchanced Regularity/Leakage-Simulation Lemma (\Cref{thm:intro-enhanced-informal}) to obtain a randomized polynomial-sized circuit $S : \zo^n\rightarrow \zo^\ell$ with probability mass function $s : \zo^n\rightarrow [0,1]^{2^\ell}$ also computable in polynomial size.

Define $Z=S(X)$.  We are promised by \Cref{thm:intro-enhanced-informal} that $(X,Y)\approx^c (X,Z)$, so we only need to argue that $\Havgmin(Z|X)\geq k$. Here we will use the fact that computational indistinguishability holds even if we provide the distinguishers with $s(X)$.  Indeed, consider the distinguisher $D : \zo^n\times \zo^\ell \times [0,1]^{2^\ell} \rightarrow \zo$ defined by:
$$D(x,y,v) =
\begin{cases}
1 & \text{if $y = \argmax_{z\in \zo^\ell} v_z$ \qquad (breaking ties in lexicographic order)}\\
0 & \text{otherwise.}
\end{cases}
$$
Observe that $\Pr[D(X,Z,s(X))=1] = 2^{-\Havgmin(Z|X)}$, because $D(X,Z,s(X))$ is exactly testing whether the optimal information-theoretic predictor of $Z$ from $X$ succeeds, since $s(X)$ is exactly the probability mass function of $Z$ given $X$.  Thus, to show that $\Havgmin(Z|X)\geq k$, we need to show that $\Pr[D(X,Z,s(X))=1]\leq 2^{-k}$.  Since $D$ is computable by a circuit of size $O(2^\ell)=\poly(n)$, we know that
$\Pr[D(X,Z,s(X))=1]$ differs only negligibly from 
$\Pr[D(X,Y,s(X))=1]$, so let's consider the latter quantity.

Observe that $\Pr[D(X,Y,s(X))=1]=\Pr[A(X)=Y]$ for
$$A(x) = \argmax_{z\in \zo^\ell} s(x)_z \qquad \text{(breaking ties in lexicographic order)}.$$
Since $s$ is computable by a polynomial-sized circuit and $\ell=O(\log n)$, $A$ is also computable by a polynomial-sized circuit.  Thus, by our hypothesis (\Cref{eqn:intro-min-rv-unpred}), $\Pr[A(X)=Y]\leq 2^{-k}+\negl(n)$. 

Thus, we conclude that 
$$2^{-\Havgmin(Z|X)} = \Pr[D(X,Z,s(X))=1] \leq 
\Pr[D(X,Y,s(X))=1] + \negl(n) \leq 2^{-k} +\negl(n),$$
exactly as desired up to the additive $\negl(n)$.  The additive $\negl(n)$ can be easily removed by a modifying $Z$ up to a negligible statistical distance.

\section{Technical Overview}
\subsection{Notation}

To describe our results and techniques more precisely, it will be convenient to switch from random-variable notation as used above to a functional notion.
Let $L$ be a positive integer, which represents $2^\ell$ in the discussion above.
Let $\Delta_L\subseteq [0,1]^L$ be the set of probability distributions over a label space $[L] := \{1,\ldots,L\}$, where we represent discrete probability distributions by their probability mass functions.
Then our random variable $(X,Y)$ can be represented by the probability distribution $\mu$ of $X$ and a function $g^* : \zo^n\rightarrow \Delta_L$, where $g^*(x)$ is the distribution of $Y|_{X=x}$.
Similarly, we can represent another random variable $Z$ jointly distributed with $X$ by a {\em simulator} $s : \zo^n\rightarrow \Delta_L$.

We normally would describe a distinguisher $d$ between $(X,Y)$ and $(X,Z)$ as a function $d : \zo^n\times [L]\rightarrow [0,1]$.  Instead, it will be convenient to describe them as functions $f : \zo^n\rightarrow [-1,1]^L$, where the $y$-th coordinate of $f(x)$ (denoted by
$f(x)_y$) represents $1-2d(x,y)$.  Then computational indistinguishability between $(X,Y)$ and $(X,Z)$ can be conveniently described as follows.
\begin{definition}
\label{def:indistinguishability}
Let $\mu$ be a distribution on $\zo^n$ and $g^*,s : \zo^n\rightarrow \Delta_L$.  We say that $g^*$ and $s$ are {\em $(T,\eps)$-indistinguishable} with respect to $\mu$ if for every $f : \zo^n\rightarrow [-1,1]^L$ with circuit complexity at most $T$, we have:
$$\left|\E_{x\sim \mu}\left[\langle s(x)-g^*(x),f(x)\rangle\right]\right|\leq \eps,$$
where $\langle \cdot, \cdot \rangle$ denotes the standard inner product on $\R^L$.
We write $\ind_\mu(g^*;T,\varepsilon)$ to denote the class of all functions $s : \zo^n\rightarrow \Delta_L$ that are $(T,\eps)$-indistinguishable from $g^*$.
\end{definition}
From now on, we fix an arbitrary distribution $\mu$ over $\zo^n$ and will often omit it from our notation.

We define our entropy measures as follows:
\begin{definition}
\label{def:entropy}
Let $\varphi : \Delta_L\rightarrow \R$ be a convex function.  For a distribution $\mu$ on $\zo^n$ and a function $g : \zo^n\rightarrow \Delta_L$, we define $$\H_\varphi(g) = -\E_{x\sim \mu}\left[\varphi(g(x))\right].$$
\end{definition}
Consider what happens when we apply this definition to a $\mu$ and $g^*$ representing a random variable $(X,Y)$ on $\zo^n\times [L]$.
If we take $\varphi(p) = \sum_b p_b\log p_b$, then $\H_\varphi(g^*) = \HSh(Y|X).$
If we take $\varphi(p) = \|p\|_\infty$, then
$\H_\varphi(g^*) = -2^{-\Havgmin(Y|X)}$, which is indeed a monotone increasing change of variables of the average min-entropy.

The divergence measures we consider are defined as follows.
\begin{definition}[Bregman Divergence]
\label{def:bregman}
Let $\varphi : \Delta_L\rightarrow \R$ be a convex function.
If $\varphi$ is differentiable, then for
$g^*,g : \zo^n\rightarrow \Delta_L$ and a distribution $\mu$ on $\zo^n$, we define the {\em
Bregman divergence} as:
\begin{align*}
\D_\varphi(g^*\|g)
&:= \E_{x\sim\mu}\left[\D_\varphi(g^*(x)\|g(x))\right]\\
&:= \E_{x\sim\mu}\left[\varphi(g^*(x)) - \varphi(g(x)) - \langle g^*(x) - g(x), \nabla \varphi(g(x))\rangle\right]\\
&= \H_\varphi(g)-\H_\varphi(g^*) -
\E_{x\sim\mu}\left[\langle g^*(x) - g(x), \nabla \varphi(g(x))\rangle\right],
\end{align*}
where $\nabla \varphi(v)\in \R^L$ is the gradient of $\varphi$ at $v\in \Delta_L$.
If $\varphi$ is not differentiable, then we replace $\nabla \varphi(p)$ with a fixed choice of {\em subgradient} of $\varphi$ at each $p\in \Delta_L$.   (See \Cref{sec:conjugate}.)
\end{definition}
Consider what happens when we apply this definition to $g^*$ and $g$ representing random variables $(X,Y)$ and $(X,A(X))$ on $\zo^n\times [L]$, respectively.
If we take $\varphi(p) = \sum_b p_b\log p_b$, then $\D_\varphi(g^*\|g) = \Dkl((X,Y)\|(X,A(X))).$
If we take $\varphi(p) = \|p\|_\infty$, this function is not differentiable and its subgradient is not unique. We choose
\[
\nabla \varphi(p) = \be_{\argmax_b p_b},
\]
where $\be_b$ is the standard basis vector whose $b$-th coordinate is one, and $\argmax_b p_b$ is the index of the largest coordinate of $p$. When there are multiple largest coordinates, we define $\argmax_b p_b$ to be the smallest index of those largest coordinates. Now we have for $v^*,v\in \Delta_L$,
\begin{align*}
\D_\varphi(v^*\|v)
&= \|v^*\|_\infty - \|v\|_\infty - \langle v^*-v, \be_{\argmax_jv_j}\rangle\\
&=
\|v^*\|_\infty - \langle v^*, \be_{\argmax_jv_j}\rangle\\
&=
\|v^*\|_\infty - v^*_{\argmax_jv_j}.
\end{align*}
Thus, we calculate the divergence between functions $g^*$ and $g$ as:
\begin{align*}
\D_\varphi(g^*\|g)
&= 2^{-\Havgmin(Y|X)}
- \Pr[Y=\argmax_b g(X)_b].
\end{align*}
Consequently, minimizing $\D_\varphi(g^*\|g)$ over all efficiently computable $g$ amounts to maximizing the probability that an efficiently computable $A'(x)=\argmax_b g(x)_b$ predicts $Y$, aligning with \Cref{thm:intro-min-rv}.

For a size bound $T \ge 0$, we use $\Time (T)$ to denote the class of functions $g:\zo^n\to \Delta_L$ with circuit complexity at most $T$. For a function $r:\Delta_L\to \R^L$, we define
\[
\|r\|_\infty:= \sup_{p\in \Delta_L}\|r(p)\|_\infty.
\]

With this notation in hand, we can restate
\Cref{thm:intro-bregman-rv} of Zheng as follows:
\begin{theorem}[Characterization of generalized pseudoentropy \cite{zheng}]
\label{thm:intro-diff}
Let $\varphi:\Delta_L \to \R$ be a differentiable strictly convex function.
Let $T_\varphi,T_\varphi'$ be size bounds such that given a linear function $\Lambda$ on $\Delta_L$, the $p\in \Delta_L$ minimizing $\varphi(p)+\Lambda(p)$ can be computed by circuits of size at most $T_\varphi$, and given $p\in \Delta_L$, the gradient $\nabla \varphi(p)$ can be computed by circuits of size at most $T_\varphi'$.
Then for every size bound $T > 0$,
every $\varepsilon\in (0,1)$,
there exists $T' = \poly(T,n,L,1/\varepsilon)+T_\varphi$ such that for every $g^*:\zo^n\to \Delta_L$,
\begin{equation}
\label{eq:intro-diff-1}
    \underbrace{\max_{s\in \ind(g^*;T,\varepsilon)}\H_\varphi(s) - \H_\varphi(g^*)}_{\text{pseudoentropy gap}} \ge \underbrace{\min_{g\in \Time(T')}\D_\varphi(g^*\|g)}_{\text{hardness of approximating}}.
\end{equation}
Conversely, for every size bound $T$
and for every $g^*:\zo^n\to \Delta_L$, we have\footnote{The dependency on $\|\nabla \varphi\|_\infty$ in \eqref{eq:intro-diff-converse} ensures that \eqref{eq:intro-diff-converse} remains invariant after scaling $\varphi$ by any positive factor.}
\begin{equation}
\label{eq:intro-diff-converse}
\underbrace{\max_{s\in \ind(g^*;T+T_\varphi',\varepsilon)}\H_\varphi(s) - \H_\varphi(g^*)}_{\text{pseudoentropy gap}} \le \underbrace{\min_{g\in \Time(T)} \D_\varphi(g^*\|g)}_{\text{hardness of approximating}} + \,\varepsilon \cdot \|\nabla \varphi\|_\infty.
\end{equation}
\end{theorem}

\Cref{thm:intro-diff} indeed implies \Cref{thm:intro-bregman-rv} by the following correspondences:
\begin{align*}
L & = 2^\ell,\\
\H_\varphi(g^*) & = \H_\varphi(Y|X),\\
\H_\varphi(s) & = \H_\varphi(Z|X),\\
\D_\varphi(g^*\|g) & = \D_\varphi((X,Y)\|(X,A(X))).
\end{align*}
The assumption $\ell = O(\log n)$ in \Cref{thm:intro-bregman-rv} gives $L = \poly(n)$. We also have $T_\varphi,T_\varphi' = \poly(n)$ from the assumptions in \Cref{thm:intro-bregman-rv}. These conditions ensure that $T' = \poly(T,n,1/\varepsilon)$ in \Cref{thm:intro-diff}, so we can use \eqref{eq:intro-diff-1} to prove that \Cref{item:bregman-2} (hardness) in \Cref{thm:intro-bregman-rv}  implies \Cref{item:bregman-1} (pseudoentropy). Similarly, the reverse implication follows from \eqref{eq:intro-diff-converse}.

\subsection{Formal Statement of Results}
We state our main results in \Cref{thm:intro-main-1,thm:intro-unif,thm:intro-unif-converse,thm:lb-exp} below.
\subsubsection{Non-uniform Setting}
\begin{theorem}
\label{thm:intro-main-1}
Let $T > 0$ be a size bound and let $\varepsilon\in (0,1)$ be an error parameter.
Let $\Phi$ be a family of convex functions $\varphi:\Delta_L\to \R$. Assume that given $v\in \Delta_L$, the subgradient $\nabla \varphi(v)$ is bounded in $[-1,1]^L$ and can be computed to $\ell_\infty$ accuracy $\eps/4$ with circuit complexity at most $T_\Phi$, for every $\varphi\in \Phi$.
Then there exists
\begin{equation}
\label{eq:blowup}
T' = O\left(\frac{(T + T_\Phi)\log L}{\varepsilon^2} + L\cdot \poly(1/\varepsilon,\log L)\right)
\end{equation}
such that for every $g^*:\zo^n\to \Delta_L$, there exists $s\in \ind(g^*;T,\varepsilon)\cap \allowbreak \Time(T')$
such that
    \begin{equation}
    \label{eq:intro-main-1}
    \H_\varphi(s) - \H_\varphi(g^*) \ge \D_\varphi(g^*\|s) - \varepsilon \quad \text{for every $\varphi\in \Phi$.}
    \end{equation}
Conversely, for every $T > 0$, there exists $T' = O(T + T_\Phi)$ such that for every $g^*:\zo^n\to \Delta_L$ and every $\varphi\in \Phi$,
\begin{equation}
\label{eq:converse}
\max_{s\in \ind(g^*;T',\varepsilon/2)}(\H_\varphi(s) - \H_\varphi(g^*)) \le \min_{g\in \Time(T)}\D_\varphi(g^*\|g) + \varepsilon.
\end{equation}
\end{theorem}
In \Cref{thm:intro-main-1}, the \emph{single} simulator $s$ satisfies \eqref{eq:intro-main-1} for \emph{all} $\varphi\in \Phi$, and it simultaneously achieves indistinguishability ($s\in\ind(g^*;T,\varepsilon)$) and computational efficiency ($s\in \Time(T')$). Therefore, \eqref{eq:intro-main-1} is stronger than the following guarantee akin to \eqref{eq:intro-diff-1} in \Cref{thm:intro-diff}:
    \[
    \max_{s\in \ind(g^*;T,\varepsilon)}\H_\varphi(s) - \H_\varphi(g^*)  \ge \min_{g\in \Time(T')} \D_\varphi(g^*\|g) - \varepsilon \quad \text{for every $\varphi\in \Phi$.}
    \]
Here, the entropy maximizer $s$ on the left (the \emph{simulator}) and the divergence minimizer $g$ on the right (the \emph{approximator}) can be different and can change based on the choice of $\varphi$. By contrast, in \eqref{eq:intro-main-1}, the single function $s$ plays the roles of the simulator and the approximator simultaneously for all $\varphi\in \Phi$.
This gives a rigorous confirmation for the aforementioned informal intuition from the Multicalibration Theorem (\Cref{thm:intro-mc-rv-informal}) that there is a close correspondence between the best efficiently computable approximator $g$ for $g^*$ (with respect to any $\D_\varphi$) and the highest-entropy simulator $s$ that is indistinguishable from $g^*$.

The increase in complexity at \eqref{eq:blowup} in \Cref{thm:intro-main-1} is polynomial in $L = 2^\ell$, a significant improvement over the doubly exponential dependence on $\ell$ in \Cref{thm:intro-mc-rv-informal}. We achieve this by focusing on a class $\Phi$ of functions $\varphi$ whose subgradients $\nabla \varphi$ are efficiently computable (up to small $\ell_\infty$ error). This efficiency assumption is naturally satisfied by common entropy notions such as the Shannon entropy ($\varphi(v) = \sum_i v_i\ln v_i$), the min entropy ($\varphi(v) = \|v\|_\infty$), and the collision probability ($\varphi(v) = \sum_i v_i^2$, see \Cref{sec:examples}). One caveat is that \Cref{thm:intro-main-1} also requires the subgradient $\nabla\varphi(v)$ to be bounded, which is not satisfied by the Shannon entropy when $v$ is close to the boundary of the probability simplex $\Delta_L$. However, this can be addressed by slightly perturbing $v$ away from the boundary of $\Delta_L$ (see \Cref{sec:unbounded}).

In \Cref{thm:lb-exp} (which we present shortly), we show that some efficiency assumption on the convex functions $\varphi$ is necessary to avoid an exponential dependence on $L$ (i.e.\ to avoid a doubly exponential dependence on $\ell$).

In contrast to \Cref{thm:intro-diff}, our \Cref{thm:intro-main-1} does not require $\varphi$ to be differentiable or strictly convex. When $\varphi$ is not differentiable, its subgradient at some $v\in \Delta_L$ may not be unique, and \Cref{thm:intro-main-1} holds as long as we use any fixed choice of subgradient to define the Bregman divergence $\D_\varphi$ (assuming that the subgradient is bounded and can be computed up to small $\ell_\infty$ error in time $T$). Therefore, our \Cref{thm:intro-main-1} encompasses both \Cref{thm:intro-min-rv,thm:intro-shannon-rv}.

By \eqref{eq:converse}, the function $s\in \ind(g^*;T,\varepsilon)\cap \allowbreak \Time(T')$ in \eqref{eq:intro-main-1} satisfies
\[
\H_\varphi(s) - \H_\varphi(g^*) \ge \max_{s'\in \ind(g^*;T'',\varepsilon/2)} (\H_\varphi(s') - \H_\varphi(g^*)) - 2\varepsilon \quad \text{for every $\varphi\in \Phi$},
\]
where
\[
T'' = O(T' + T_\Phi) = O(T') = O\left(\frac{(T + T_\Phi)\log L}{\varepsilon^2} + L\cdot \poly(1/\varepsilon,\log L)\right).
\]
This means that the \emph{single} indistinguishable function $s\in \ind(g^*;T,\varepsilon)$ achieves comparable or higher entropy than \emph{every} $s'\in \ind(g^*;T'',\varepsilon/2)$ w.r.t.\ \emph{every} entropy notion $\H_\varphi$ for $\varphi\in \Phi$. In this sense, $s$ is a \emph{universal simulator} for the pseudoentropy of $g^*$.

Combining the two directions \eqref{eq:intro-main-1} and \eqref{eq:converse} of \Cref{thm:intro-main-1} in a different order, we get that for every size bound $T > 0$ and every $g^*:\zo^n\to \Delta_L$, there exist $T' = O(T + T_\Phi)$,
\[
T'' = O\left(\frac{(T' + T_\Phi)\log L}{\varepsilon^2} + L\cdot \poly(1/\varepsilon,\log L)\right)
\]
and a function $s\in \ind(g^*;T',\varepsilon/2)\cap \Time(T'')$ such that
\[
\D_\varphi(g^*\|s) \le \min_{g\in \Time(T)}\D_\varphi(g^*\|g) + 2\varepsilon \quad \text{for every $\varphi\in \Phi$}.
\]
This means that the \emph{single} low-complexity function $s\in \Time (T'')$ achieves comparable or better approximation error than \emph{every} $g\in \Time (T)$ w.r.t.\ \emph{every} divergence notion $\D_\varphi$ for $\varphi\in \Phi$.
In the recent learning theory literature, a function $s$ satisfying this property is termed an \emph{omnipredictor} \cite{omni}. Our \Cref{thm:intro-main-1} thus recovers a result of \cite{loss-oi} showing a construction of omnipredictors based on \emph{calibrated multiaccuracy}. (See \Cref{sec:proof-idea} for more discussion.)

We discuss our proof ideas for \Cref{thm:intro-main-1} in \Cref{sec:proof-idea} and present the formal proof in \Cref{sec:proof-main-non-unif}.

\subsubsection{Uniform Setting}

In \Cref{thm:intro-unif} below, we prove a uniform version of \Cref{thm:intro-main-1} that encompasses the uniform versions of \Cref{thm:intro-min-rv,thm:intro-shannon-rv} in \cite{characterize,zheng} and extends \Cref{thm:intro-bregman-rv} also to the uniform setting (which has not been explicitly achieved in prior work).
In the uniform setting, instead of having a single function $g^*:\zo^n\to \Delta_L$, we are interested in a family of functions $g^*_{n,L}:\zo^n\to \Delta_L$ defined for varying choices of input length $n$ and label space size $L$. \Cref{thm:intro-main-1} only allows us to construct an efficiently computable function $s$ satisfying \eqref{eq:intro-main-1} separately for each choice of $(n,L)$, so it does not guarantee the existence of a single uniform algorithm (e.g.\ a Turing machine) that efficiently computes $s$ for all choices of $(n,L)$. Our \Cref{thm:intro-unif} addresses this limitation.

To state \Cref{thm:intro-unif}, we use terminology similar to the uniform versions of \Cref{thm:intro-min-rv,thm:intro-shannon-rv} in \cite{characterize,zheng}, where we formalize a uniform distinguisher as an efficient uniform oracle $\cA$ that, when given sample access to $g^*,g:\zo^n\to \Delta_L$ on a distribution $\mu$,\footnote{This means that $\cA$ can obtain random samples $(x,y,y^*)\in \zo^n\times [L]\times [L]$, where we first draw $x$ from $\mu$ and then independently draw $y\sim g(x),y^*\sim g^*(x)$.} aims to output a function $f:\zo^n\to [-1,1]^L$ that witnesses the violation of indistinguishability:
\begin{equation}
\label{eq:intro-unif-ind}
\langle g - g^*, f \rangle > \varepsilon.
\end{equation}
We say $g$ is $(\cA,\varepsilon)$-\emph{distinguishable} from $g^*$ if \eqref{eq:intro-unif-ind} holds with high probability, and otherwise we say  $g$ is $(\cA,\varepsilon)$-indistinguishable. We defer the formal definition to \Cref{def:dist-calib-oracles}.

Our goal is to construct an $(\cA,\varepsilon)$-indistinguishable function $s$ satisfying \eqref{eq:intro-main-1} using an efficient uniform algorithm that takes $\cA$ and $\cB$ as subroutines, and we still want $s$ to satisfy \eqref{eq:intro-main-1} simultaneously for all $\varphi\in \Phi$. To achieve this, we assume that the function class $\{\nabla \varphi\}_{\varphi\in \Phi}$ is \emph{weakly-agnostically learnable} by $\cB$,
which is trivially satisfied when $\Phi$ only contains a single $\varphi$ whose subgradient $\nabla \varphi$ can be efficiently approximated (as in \Cref{thm:intro-min-rv,thm:intro-shannon-rv}). Roughly speaking, we assume that $\cB$ can weakly agnostically find a calibration function $r(v)$ whenever some $r(v)=\nabla\varphi(v)$ has large correlation with $g-g^*$. In the proof, we combine $\cA$ and $\cB$ into a single general distinguishing oracle. We defer the formal definition to \Cref{def:weak-ag}.

\begin{theorem}[Uniform version of \Cref{thm:intro-main-1}; informal statement of \Cref{thm:unif}]
\label{thm:intro-unif}
Let $\Phi$ be a class of convex functions $\varphi:\Delta_L\to \R$, and let $\mu$ be an arbitrary distribution on $\zo^n$.
    Let $\cA$ be an arbitrary distinguishing oracle, and let $\cB$ be an $\varepsilon_1$-weak agnostic calibration oracle for the family $\{\nabla \varphi\}_{\varphi\in\Phi}$. We can use $\cA$ and $\cB$ as subroutines to construct an efficient (uniform) algorithm $\cS$ that, given sample access to $g^*:\zo^n\to \Delta_L$ on $\mu$, with high probability outputs a (low-complexity) function $s$ that is $(\cA,\varepsilon)$-indistinguishable from $g^*$, and such that
    \[
    \H_\varphi(s) - \H_\varphi(g^*) \ge \D_\varphi(g^*\|s) - \varepsilon_1  \quad \text{for every $\varphi\in \Phi$}.
    \]
\end{theorem}

\Cref{thm:intro-unif} encompasses the uniform versions of \Cref{thm:intro-min-rv,thm:intro-shannon-rv} in \cite{characterize,zheng} by choosing the entropy notion $\H_\varphi$ accordingly. In the uniform versions of \Cref{thm:intro-min-rv,thm:intro-shannon-rv}, we assume that every function $s'$ with $\H_\varphi(s') - \H_\varphi(g^*)$ exceeding some threshold $\beta$ is $(\cA,\varepsilon)$-distinguishable from $g^*$. Under this assumption, since the function $s$ from \Cref{thm:intro-unif} is $(\cA,\varepsilon)$-indistinguishable from $g^*$, it must hold that $\H_\varphi(s) - \H_\varphi(g^*) \le \beta$, which implies that $\D_\varphi(g^*\|s) \le \beta + \varepsilon_1$. In this case, the algorithm $\cS$ from \Cref{thm:intro-unif} becomes an efficient uniform algorithm that approximates $g^*$ within low $\D_\varphi$ error, exactly as needed to prove the uniform versions of \Cref{thm:intro-min-rv,thm:intro-shannon-rv}.

We remark that the uniform versions of \Cref{thm:intro-min-rv,thm:intro-shannon-rv} in \cite{characterize,zheng} also have a reverse direction, which extends to general entropy notions $\H_\varphi$ and gives a uniform version of the reverse direction \eqref{eq:converse} of \Cref{thm:intro-main-1}. This generalization can be proved straightforwardly using the same proof idea as its nonuniform counterpart. We state this result as follows:

\begin{theorem}[Reverse direction of \Cref{thm:intro-unif}]
\label{thm:intro-unif-converse}
Let $\cA$ be an arbitrary uniform-time-$T$ oracle computing a function $g:\zo^n\to\Delta_L$.
Let $\varphi:\Delta_L\to \R$ be a convex function. Assume, similarly to \Cref{thm:intro-main-1},  that given $v\in \Delta_L$, the subgradient $\nabla \varphi(v)$ is bounded in $[-1,1]^L$ and can be computed to $\ell_\infty$ accuracy $\eps/4$ in uniform time $T_\Phi$ using an oracle $\cB$.
There exists an $O(T + T_\Phi)$-uniform-time algorithm that computes a function $f:\zo^n\to[-1,1]^L$ using $\cA$ and $\cB$ as subroutines with the following property.
For every distribution $\mu$ on $\zo^n$, every function $g^*:\zo^n\to\Delta_L$,
and every function $s:\zo^n\to\Delta_L$ that is $(\varepsilon/2)$-indistinguishable from $g^*$ by $f$, meaning
\[
\langle s-g^*,f\rangle\le \varepsilon/2,
\]
it holds
\[
\H_\varphi(s)-\H_\varphi(g^*) \le \D_\varphi(g^*\|g)+\varepsilon.
\]
\end{theorem}
We present the proofs of \Cref{thm:intro-unif,thm:intro-unif-converse} in \Cref{sec:proof-main-unif}.

\subsubsection{Exponential Lower Bound}
In our \Cref{thm:intro-main-1}, we make a computational efficiency assumption about the convex functions $\varphi\in \Phi$ corresponding to the entropy notions considered in the theorem: we assume that the subgradient function $\nabla \varphi$ can be approximated by a circuit of size at most $T_\Phi$ for every $\varphi\in \Phi$. Correspondingly, the circuit complexity of the function $s$ guaranteed by the theorem depends polynomially on $T_\Phi$, in addition to the polynomial dependency on $T,L$ and $1/\varepsilon$.

In this section, we prove the following lower-bound theorem showing that some complexity assumption on $\varphi$ is necessary to avoid an exponential dependency on $L$ in the circuit complexity of $s$. This lower bound holds even when $\Phi$ only contains a single convex function $\varphi$ and there is only a single, extremely simple distinguisher $f$:
\begin{theorem}
\label{thm:lb-exp}
    For every sufficiently large positive integer $n$, choosing $L = n$, there exist a distribution $\mu$ over $\zo^n$, a function $g^*:\zo^n\to \Delta_L$, an $O(n)$-sized circuit $f:\zo^n\to [-1,1]^L$, and a convex function $\varphi:\Delta_L\to \R$ with bounded subgradient $\nabla \varphi(v)\in [-1,1]^L$ for every $v\in \Delta_L$ that satisfy the following property. Let $s:\zo^n\to \Delta_L$ be an arbitrary function that is $(\{f\},0.05)$-indistinguishable from $g^*$ and satisfies the following inequality:
\begin{equation}
\label{eq:larger-entropy}
\H_\varphi(s) \ge \H_\varphi(g^*) - 0.05.
\end{equation}
    Then $s$ must have circuit complexity $\exp(\Omega(n)) = \exp(\Omega(L))$.
    Moreover, we can simply choose $f$ to be the identity function: $f(x) = x\in \zo^n\subseteq[-1,1]^L$ for every $x\in \zo^n$.
\end{theorem}
The condition \eqref{eq:larger-entropy} is very mild. In particular, by the nonnegativity of Bregman divergences, \eqref{eq:larger-entropy} is a necessary condition for the following guarantee of our main theorem (\Cref{thm:intro-main-1}):
\begin{equation}
\label{eq:guarantee-main}
\H_\varphi(s) - \H_\varphi(g^*) \ge \D_\varphi(g^*\|s) - 0.05.
\end{equation}
Therefore, \Cref{thm:lb-exp} implies an exponential lower bound on the circuit complexity of every function $s\in \ind(g^*;\{f\},0.05)$ that satisfies the guarantee \eqref{eq:intro-main-1} of \Cref{thm:intro-main-1}, even when $\Phi$ only contains a single function $\varphi$ and there is only a single, extremely simple distinguisher $f$ (the identity function).

There is no contradiction between the $\exp(\Omega(L))$ lower bound (\Cref{thm:lb-exp}) and the $\poly(L)$ upper bound (\Cref{thm:intro-main-1}) on the circuit complexity of $s$ because the upper bound assumes that the convex function $\varphi$ has low complexity (specifically, its subgradient can be approximated by a circuit of size $T_\Phi$). Thus, \Cref{thm:lb-exp} shows that some assumptions on the complexity of $\varphi$ are necessary for \Cref{thm:intro-main-1} to hold, even when $\Phi$ only contains a single convex function $\varphi$.

Our proof of \Cref{thm:lb-exp} uses coding-theoretic ideas (specifically, the existence of a sufficiently large \emph{design}) as well as a probabilistic counting argument. See \Cref{sec:lb} for the formal proof.

As a corollary of \Cref{thm:lb-exp}, we give an exponential lower bound on the circuit complexity for achieving multiaccuracy (i.e.\ indistinguishability) plus weight-restricted calibration, even for a single, extremely simple distinguisher $f$ (the identity function) and a single weight function $\nabla \varphi:\Delta_L\to [-1,1]^L$. Consequently, this lower bound also holds for stronger notions such as calibrated multiaccuracy and multicalibration. See \Cref{sec:proof-idea} for a more detailed discussion about these notions from the algorithmic fairness literature (multiaccuracy, weight-restricted calibration, calibrated multiaccuracy, and multicalibration).

\begin{corollary}
\label{cor:lb}
    For every sufficiently large positive integer $n$, choosing $L = n$, the same distribution $\mu$, function $g^*$, distinguisher $f$ and convex function $\varphi$ from \Cref{thm:lb-exp} have the following property. Let $s:\zo^n\to \Delta_L$ be an arbitrary function satisfying the following two conditions:
    \begin{align*}
        |\E_{x\sim\mu}\langle s(x) - g^*(x), f(x)\rangle| & \le 0.05, \tag{\textbf{Indistinguishability, a.k.a.\ multiaccuracy}}\\
        \E_{x\sim\mu}\langle s(x) - g^*(x),\nabla \varphi (s(x)) \rangle & \le 0.05. \tag{\textbf{Weight-restricted calibration}}
    \end{align*}
Then $s$ must have circuit complexity $\exp(\Omega(n)) = \exp(\Omega(L))$.
\end{corollary}
\subsection{Proof Idea: Enhanced Regularity/Leakage-Simulation Lemma with Weight-Restricted Calibration}
\label{sec:proof-idea}

We now give a high-level description of our proof idea for the first part of \Cref{thm:intro-main-1}, i.e.\ the forward direction \eqref{eq:intro-main-1}. Our proof of the reverse direction \eqref{eq:converse} follows the same ideas as in prior work \cite{characterize,zheng}.
We present the formal proof in \Cref{sec:proof-main-non-unif}.

Our idea is simple in hindsight. As mentioned earlier, a key difference between our guarantee \eqref{eq:intro-main-1} and previous ones (e.g.\ \eqref{eq:intro-diff-1}) is that the same function $s$ appears on both sides of the inequality, playing the roles of the simulator and the approximator simultaneously. This makes it natural to calculate the difference between the two sides. Specifically, using the definitions of $\H_\varphi$ and $\D_\varphi$ in \Cref{def:entropy,def:bregman}, we obtain the following identity, which is the central technical insight of our work:

\begin{equation}
\label{eq:gap-calibration}
\underbrace{\H_\varphi(s) - \H_\varphi(g^*)}_{\text{pseudoentropy gap}}
- \underbrace{\D_\varphi(g^*\|s)}_{\text{hardness of approximating}}
= \langle g^* - s,  \nabla \varphi \circ s \rangle.
\end{equation}
In the right-hand side of \eqref{eq:gap-calibration} and for the rest of the paper, whenever a distribution $\mu$ on $\zo^n$ is clear from context, we define the inner product of any two functions $g,h:\zo^n\to \R^L$ with respect to $\mu$ as
\[
\langle g,h\rangle:= \E_{x\sim\mu}\langle g(x),h(x) \rangle.
\]

By \eqref{eq:gap-calibration}, our goal \eqref{eq:intro-main-1} simplifies to constructing $s$ satisfying
\begin{equation}
\label{eq:proof-idea-1}
\langle s - g^*,\nabla \varphi\circ s \rangle \le \varepsilon \quad \text{for every $\varphi\in \Phi$},
\end{equation}
such that $s$ is efficiently computable and indistinguishable from $g^*$.
In \Cref{thm:intro-main-1}, we assume that $\nabla \varphi$ is bounded in $[-1,1]^L$ and can be approximated within $\ell_\infty$ error $\varepsilon/4$ by circuits of size at most $T_\Phi$. Thus for every $\varphi\in \Phi$, there exists $r_\varphi:\Delta_L\to [-1,1]^L$ of circuit complexity at most $T_\Phi$ such that $\|r_\varphi - \nabla\varphi\|_\infty \le \varepsilon/4$. This gives a sufficient condition for \eqref{eq:proof-idea-1}:
\begin{equation}
\label{eq:proof-idea-2}
\langle s - g^*, r_\varphi\circ s \rangle \le \varepsilon/2 \quad \text{for every $\varphi\in \Phi$}. \quad \textbf{(Weight-restricted calibration)}
\end{equation}
As we will explain shortly, the condition \eqref{eq:proof-idea-2} is termed \emph{weight-restricted calibration} in the recent literature on algorithmic fairness.
On the other hand, the requirement that $s$ is indistinguishable from $g^*$ (\Cref{def:indistinguishability}) amounts to
\begin{equation}
\label{eq:proof-idea-3}
|\langle s - g^*, f\rangle| \le \varepsilon \quad \text{for every $f\in \cF$}, \quad \textbf{(Indistinguishability, a.k.a.\ multiaccuracy)}
\end{equation}
where $\cF$ denotes the class of functions $f:\zo^n\to [-1,1]^L$ with circuit complexity at most $T$. Again, as we will explain below, the indistinguishability condition \eqref{eq:proof-idea-3} is termed \emph{multiaccuracy} in the recent literature on algorithmic fairness.

The first half of \Cref{thm:intro-main-1} requires us to find a function $s\in \ind(g^*;T,\varepsilon)\cap \Time(T')$ satisfying \eqref{eq:intro-main-1}. Based on the discussion above, the requirement $s\in \ind(g^*;T,\varepsilon)$ is equivalent to \eqref{eq:proof-idea-3}, and a sufficient condition for \eqref{eq:intro-main-1} is given by \eqref{eq:proof-idea-2}. Thus our goal becomes finding $s\in \Time(T')$ satisfying both \eqref{eq:proof-idea-2} and \eqref{eq:proof-idea-3}.

The structure of \eqref{eq:proof-idea-2} and \eqref{eq:proof-idea-3} points us to previous ideas about \emph{regularity/leakage-simulation lemmas}.
Specifically, the Complexity-Theoretic Regularity Lemma of Trevisan, Tulsiani, and Vadhan \cite{ttv} and its extension to $L > 2$ in the Leakage Simulation Lemma of Jetchev and Pietrzak \cite{leakage}, allow us to construct $s\in \Time (T')$ that satisfies \eqref{eq:proof-idea-3} alone. This result can be proved using a boosting algorithm. We start with a simple $s:\zo^n\to \Delta_L$ (e.g.\ $s(x)$ is the uniform distribution over $[L]$ for every $x\in \zo^n$). Whenever $s$ does not satisfy \eqref{eq:proof-idea-3}, there exists a (low-complexity) $f\in \cF$ that witnesses this violation, which we use to update $s$ to reduce a nonnegative potential function by a significant amount. By the nonnegativity of the potential function, there can only be a bounded number of updates, after which $s$ has to satisfy \eqref{eq:proof-idea-3}.
Since each $f\in \cF$ has circuit complexity at most $T$, each update increases the complexity of $s$ only by a bounded amount, and since the number of total updates is also bounded, we can ensure that the final $s$ belongs to $\Time(T')$.
Following the work of Vadhan and Zheng \cite{unif-minmax} which proves a uniform version of the regularity/leakage-simulation lemma, we apply the classic Multiplicative Weights algorithm (a special case of Mirror Descent) to extend the boosting idea from \cite{ttv} to general $L$ (see \Cref{sec:mirror,sec:mw}).

A limitation of these regularity/leakage-simulation lemmas is that they only establish \eqref{eq:proof-idea-3}, but we need \eqref{eq:proof-idea-2} and \eqref{eq:proof-idea-3} simultaneously.
The two conditions \eqref{eq:proof-idea-2} and \eqref{eq:proof-idea-3} look very similar, with both $r_\varphi$ in \eqref{eq:proof-idea-2} and $f$ in \eqref{eq:proof-idea-3} having bounded circuit complexity.
The difference is that \eqref{eq:proof-idea-2} is \emph{circular}: the ``distinguishing function'' $r_\varphi\circ s$ that $s$ tries to ``fool'' depends on the function $s$ itself.

A key insight from the recent algorithmic fairness literature about \emph{multicalibration} \cite{hkrr,low-degree-cal} is that the circularity in \eqref{eq:proof-idea-2} does not incur any new major technical challenge: the same boosting idea from the regularity/leakage-simulation lemma allows us to construct $s\in \Time(T')$ satisfying both \eqref{eq:proof-idea-2} and \eqref{eq:proof-idea-3}, which is exactly what we need to establish \Cref{thm:intro-main-1}.
We view condition \eqref{eq:proof-idea-2} as an instantiation of \emph{weight-restricted calibration}, as introduced by Gopalan, Kim, Singhal, and Zhao \cite{low-degree-cal}.\footnote{They called it \emph{weighted calibration} (see also \cite{GHR}). We modify the terminology to more clearly convey that it is a weaker condition than calibration.}
This is a relaxed version of the following condition called \emph{calibration}:
\begin{equation}
\label{eq:cal}
\E_{x\sim\mu}[s(x) - g^*(x)|s(x)] \approx \mathbf 0, \quad \textbf{(Calibration)}
\end{equation}
where $\mathbf 0\in \R^L$ is the zero vector.
Indeed, calibration \eqref{eq:cal} is equivalent to
\begin{equation}
\label{eq:cal-equiv}
\langle s - g^*, r\circ s\rangle \approx 0 \quad \text{for \emph{all} $r:\Delta_L\to [-1,1]^L$}, \quad \textbf{(Equivalent def.\ of calibration)}
\end{equation}
where each $r:\Delta_L\to [-1,1]^L$ is termed a \emph{weight function} \cite{low-degree-cal}. The weight-restricted calibration condition \eqref{eq:proof-idea-2} is a relaxation of (full) calibration \eqref{eq:cal-equiv} in that it only considers a restricted family of weight functions $\{r_\varphi\}_{\varphi\in \Phi}$ instead of all weight functions (and it replaces ``$\approx 0$'' with ``$\le \varepsilon$'').

As pointed out by Dwork, Lee, Lin, and Tankala \cite{dllt},
the indistinguishability condition \eqref{eq:proof-idea-3} (previously studied as in the Complexity-Theoretic Regularity Lemma \cite{ttv})
was rediscovered in the algorithmic fairness literature under the term \emph{multiaccuracy} \cite{hkrr,ma} and it is natural to ask whether stronger notions from the algorithmic fairness literature have implications in complexity theory and other domains where similar regularity lemmas are used (such as graph theory and cryptography).
This question has been explored by many works in recent research \cite{generative-huge,dllt,cdv}.
Our work (especially the observation \eqref{eq:gap-calibration}) demonstrates the power of adding weight-restricted calibration \eqref{eq:proof-idea-2} to the standard indistinguishability/multiaccuracy condition \eqref{eq:proof-idea-3} guaranteed by the regularity/leakage-simulation lemma for generalizing and strengthening pseudoentropy characterizations.

The combination of multiaccuracy \eqref{eq:proof-idea-3} and (full) calibration \eqref{eq:cal} (or \eqref{eq:cal-equiv}) is termed \emph{calibrated multiaccuracy} in \cite{loss-oi}. That work uses calibrated multiaccuracy to achieve omni\-prediction more efficiently than the \emph{multicalibration}-based approach of \cite{omni} (we discuss multi\-calibration in the next paragraph). The notion of multiaccuracy \eqref{eq:proof-idea-3} plus weight-restricted calibration \eqref{eq:proof-idea-2} (as we use in our work) is weaker than calibrated multiaccuracy and suffices to ensure omniprediction in many cases, leading to further improved sample complexity and computational efficiency even when $L = 2$ \cite{omni-gap,optimal-omni}.
Independent work \cite{global-cal-ma,Tankala25} demonstrates the power of (variants of) calibrated multiaccuracy for other cryptographic, complexity-theoretic, and learning-theoretic applications (as we discussed in \Cref{sec:contributions}).

As mentioned earlier, the work of Casacuberta, Dwork, and Vadhan \cite{cdv} uses \emph{multicalibration} to establish a characterization for pseudo-min-entropy and pseudo-Shannon-entropy. Introduced by \cite{hkrr}, multicalibration is the following much stronger condition than \eqref{eq:proof-idea-2}, \eqref{eq:proof-idea-3} and \eqref{eq:cal} combined:
\[
\E_{x\sim \mu}[\langle s(x) - g^*(x), f(x)\rangle|s(x)] \approx 0 \quad \text{for every $f\in \cF$}. \quad \textbf{(Multicalibration)}
\]
A contribution of our work is to show that multicalibration is stronger than necessary for pseudoentropy characterizations and leads to the unnecessary doubly exponential dependence on $\ell$ in \Cref{thm:intro-mc-rv-informal} (our \Cref{cor:lb} shows that this doubly exponential dependence cannot be avoided for multicalibration and even for the weaker notion of calibrated multiaccuracy).
Based on our observation \eqref{eq:gap-calibration}, the much weaker condition of multiaccuracy \eqref{eq:proof-idea-3} plus weight-restricted calibration \eqref{eq:proof-idea-2} is sufficient even for more general entropy notions $\H_\varphi$, giving the improved polynomial dependence on $L = 2^\ell$ in \Cref{thm:intro-main-1}.

\subsection{Enhanced Regularity/Leakage-Simulation Lemma in Functional Notation}

We obtain our combination of multiaccuracy and weight-restricted calibration by considering a more general form of distinguisher
$\tau:\zo^n\times\Delta_L\to [-1,1]^L$ that takes {\em both} an input $x\in \zo^n$ and a label distribution $v\in \Delta_L$.
Given such a $\tau$ and a function $s:\zo^n\to \Delta_L$, we write $\tau_s:\zo^n\to [-1,1]^L$ for the function $\tau_s(x):=\tau(x,s(x))$.


\begin{theorem}[Enhanced Regularity/Leakage-Simulation Lemma for General Distinguishers]
\label{thm:intro-reg-non-unif}
    Let $\gendistfamily_T$ consist of all distinguishers $\tau:\zo^n\times \Delta_L\to [-1,1]^L$ with circuit complexity at most $T$. For every distribution $\mu$ on $\zo^n$, every $g^*:\zo^n\to \Delta_L$, and every $\varepsilon\in (0,1/2)$, there exists a simulator $s:\zo^n\to \Delta_L$ such that
    \begin{enumerate}
        \item $s$ has circuit complexity
        \[
        O\left(\frac{T\log L}{\varepsilon^2} + L\cdot \poly(1/\varepsilon,\log L)\right);
        \]
        \item $s$ is indistinguishable from $g^*$ w.r.t.\ $\gendistfamily_T$:
        \[
        \langle s - g^*,\tau_s\rangle \le \varepsilon\quad \text{for every $\tau\in \gendistfamily_T$}.
        \]
    \end{enumerate}
\end{theorem}

Notice that if we restrict to a class of functions $\tau(x,v)=f(x)$ that depend only on their first argument, we get exactly the multiaccuracy condition \eqref{eq:proof-idea-3}. On the other hand, if we restrict to a class of functions $\tau(x,v)=r(v)$ that depend only on their second argument, we get exactly the weight-restricted calibration condition \eqref{eq:proof-idea-2}.  Thus, the indistinguishability condition in \Cref{thm:intro-reg-non-unif} captures both conditions and more.

Translating \Cref{thm:intro-reg-non-unif} to random-variable notation gives precisely \Cref{thm:intro-reg-non-unif-rv} stated earlier.

\subsection{Examples beyond Shannon and Min Entropy}
\label{sec:examples}
Our pseudoentropy characterization (\Cref{thm:intro-main-1}) holds for a general family $\Phi$ of entropy notions. Here we give two examples of entropy notions beyond Shannon and min entropy.

Consider the following choice of $\varphi$ with its corresponding Bregman divergence:
\begin{equation}
\label{eq:example}
    \varphi(v) = \|v\|_2^2 = \sum_{i\in [L]}v_i^2 , \quad \D_\varphi(u\|v) = \|u - v\|_2^2 = \|u\|_2^2+\|v\|_2^2 - 2\langle u,v\rangle.
\end{equation}
To interpret the entropy and divergence notions from this choice of $\varphi$, let us consider a pair of jointly distributed random variables $(X,Y)\in \zo^n\times [L]$ with $X$ drawn from a distribution $\mu$ over $\zo^n$, and define $g^*:\zo^n\to \Delta_L$ so that $g^*(x)$ is the conditional distribution of $Y$ given $X = x$. In this case, the entropy notion $\H_\varphi$ is the negative conditional collision probability $\CP(g^*(X)|X)$ used, for example, in \cite{ChungMiVa} for the analysis of hash functions:
\[
\H_\varphi(g^*) = -\E_{x\sim \mu}[\varphi(g^*(x))] = -\CP(g^*(X)|X).
\]
Here,
\[
\CP(g^*(X)|X) := \Pr[Y' = Y],
\]
where $(X,Y')$ is identically distributed as $(X,Y)$, and $Y'$ is conditionally independent from $Y$ given $X$.

By \eqref{eq:example}, for a fixed $g^*:\zo^n \to \Delta_L$, minimizing the Bregman divergence $\D_\varphi(g^*\|g)$ over $g:\zo^n\to \Delta_L$ corresponds to maximizing
\begin{equation}
\label{eq:example-2}
2\langle g^*,g\rangle - \E_{x\sim \mu}\|g(x)\|_2^2 = 2\Pr[Y = \hat Y] - \CP(g(X)|X),
\end{equation}
where the first term $\Pr[Y = \hat Y]$ is the prediction accuracy of $g$ w.r.t.\ $g^*$ (the same as in \Cref{thm:intro-min-rv}): here $X\sim \mu$, and conditioned on $X$, we independently draw $Y\sim g^*(X)$ and $\hat Y\sim g(X)$.
The second term on the right-hand side of \eqref{eq:example-2} can be viewed as a regularizer that encourages $g$ to have lower collision probability (and thus higher entropy $\H_\varphi$).
Our characterization \Cref{thm:intro-main-1} implies that the best that an efficient $g$ can do for this regularized optimization task is equivalent to the conditional pseudo-collision-probability of $g^*(X)$ given $X$ (i.e., the pseudo / computational analogue of $\CP(g^*(X)|X)$).

Similarly, we can take
\[
\varphi(v) = \|v\|_2 = \sqrt{\sum_i v_i^2},  \quad \D_\varphi(u\|v) = \frac{\|u\|_2\|v\|_2 - \langle u,v\rangle}{\|v\|_2}
\]
to get an equivalence between maximizing the following objective over efficient $g$
\[
\Pr[Y=\hat Y]/\CP^{1/2}(g(X)|X),
\]
and the conditional pseudo-square-root-collision-probability of $g^*(X)$ given $X$ (i.e., the pseudo / computational analogue of $\CP^{1/2}(g^*(X)|X)$ used, for example, in \cite{HaitnerNgOnReVa} to construct statistically hiding commitments
and statistical zero-knowledge arguments).

\section{Preliminaries}

For simplicity, we consider general and abstract vectors $g,h,f\in \R^L$ in this section. Later we will instantiate the vector $g$ here as the function value $g(x)\in \Delta_L\subseteq \R^L$ (as in \Cref{def:indistinguishability}) assigned to a specific input $x\in \zo^n$, and similarly instantiate the vector $f$ here as the function value $f(x)\in [-1,1]^L$ of a distinguisher $f:\zo^n\to [-1,1]^L$ on a specific input $x\in \zo^n$.

\subsection{Convex Analysis Basics}
\label{sec:conjugate}
Let $\cG\subseteq \R^L$ be a non-empty bounded convex set. Let $\varphi:\cG\to\R$ be a convex function on $\cG$. We say $h\in \R^L$ is a \emph{subgradient} of $\varphi$ at $g\in \cG$ if
\[
\varphi(g') - \varphi(g) \ge \langle g' - g, h \rangle \quad \text{for every $g'\in \cG$},
\]
or equivalently,
\begin{equation}
\label{eq:subgradient}
\langle g,h\rangle - \varphi(g) = \max_{g'\in \cG}(\langle g',h \rangle - \varphi(g')).
\end{equation}
The \emph{convex conjugate} of $\varphi$, denoted by $\psi:\R^L \to \R$, is defined as follows:
\begin{equation}
\label{eq:conjugate}
\psi(h) := \sup_{g\in \cG}(\langle g,h\rangle - \varphi(g)) \quad \text{for every $h\in \R^L$}.
\end{equation}
By the boundedness of $\cG$ and the convexity of $\varphi$, the supremum above is always well-defined (i.e., it is neither $+\infty$ nor $-\infty$).
It is easy to verify that $\psi$ is convex and closed.\footnote{
As standard terminology in convex analysis, we say $\psi$ is \emph{closed} if its epigraph $\{(h,y)\in \R^L\times \R:y \ge \psi(h)\}$ is closed (see e.g.\ \cite{conv-ana}).
} The \emph{Fenchel-Young divergence} $\Gamma_{\varphi,\psi}:\cG\times \R^L\to \R$ is defined as follows:
\[
\Gamma_{\varphi,\psi}(g,h) := \varphi(g) + \psi(h) - \langle g, h\rangle \ge 0 \quad \text{for every $g\in \cG$ and $h\in \R^L$}.
\]
The nonnegativity of the divergence follows immediately from the definition of $\psi$ in \eqref{eq:conjugate}. We will often omit the subscripts and simply write $\Gamma$ in place of $\Gamma_{\varphi,\psi}$.  By the characterization of subgradients in \eqref{eq:subgradient}, for any $g\in \cG$ and $h\in \R^L$,
\begin{equation}
\label{eq:conjugate-implications}
\text{$h$ is a subgradient of $\varphi$ at $g$} \quad \Longleftrightarrow \quad \Gamma (g,h) = 0 \quad \Longrightarrow \quad \text{$g$ is a subgradient of $\psi$ at $h$}.
\end{equation}
Note that the converse of the second implication is not generally true when $\varphi$ is discontinuous at the boundary of $\cG$. We say a pair $(g,h)\in \cG\times \R^L$ is a \emph{conjugate pair} w.r.t.\ $(\varphi,\psi)$ if $
\Gamma_{\varphi,\psi}(g,h) = 0.
$

The subgradient of $\varphi$ at any $g\in \cG$ may not be unique (and it may not exist when $g$ is at the boundary of $\cG$). When we have a predetermined choice $\nabla \varphi(g)$ among the subgradients, we define the \emph{Bregman divergence} $\D_\varphi(g'\|g)$ as follows:
\begin{align*}
\D_\varphi (g'\|g) := \Gamma(g', \nabla \varphi(g)) & = \varphi(g') + \psi(\nabla \varphi(g)) - \langle g', \nabla \varphi(g)\rangle\\
& = \varphi(g') - \varphi(g) - \langle g' - g, \nabla \varphi(g) \rangle,
\end{align*}
where the last equation uses the fact that $(g,\nabla \varphi(g))$ is a conjugate pair (by the equivalence in \eqref{eq:conjugate-implications}):
\[
0 = \Gamma(g,\nabla \varphi(g)) = \varphi(g) + \psi(\nabla \varphi(g)) - \langle g, \nabla \varphi(g)\rangle.
\]

If we additionally assume that $\varphi$ is closed (which holds, for example, when $\varphi$ is continuous on a closed domain $\cG$), the supremum in \eqref{eq:conjugate} can always be attained by some $g\in \cG$.
Moreover, it is a standard result that for a closed convex function $\varphi$, taking convex conjugate twice gives back the same function. That is, the closedness assumption of $\varphi$ ensures that $\varphi$ and $\psi$ are convex conjugates of each other:
\begin{align}
\psi(h) & = \max_{g\in \cG}(\langle g,h\rangle - \varphi(g)) \quad \text{for every $h\in \R^L$},\label{eq:conjugate-1}\\
\varphi(g) & = \sup_{h\in \R^L}(\langle g,h\rangle - \psi(h)) \quad \text{for every $g\in \cG$}.\label{eq:conjugate-2}
\end{align}
In this case, by the characterization of subgradients in \eqref{eq:subgradient}, for any $g\in \cG$ and $h\in \R^L$, the following three statements are equivalent:
\begin{enumerate}
    \item $\Gamma(g,h) = 0$, i.e., $(g,h)$ is a conjugate pair;
    \item $h$ is a subgradient of $\varphi$ at $g$;
    \item  $g$ is a subgradient of $\psi$ at $h$.
\end{enumerate}
\subsection{Mirror Descent}
\label{sec:mirror}
As in the previous section, we consider a non-empty bounded convex set $\cG\subseteq \R^L$ and a convex function $\varphi$ defined on it, with $\psi$ being the convex conjugate of $\varphi$.
Now we additionally assume that $\psi$ is smooth:
\begin{definition}[Smoothness]
    For $\lambda \ge 0$ and a norm $\|\cdot \|$ on $\R^L$, we say a convex function $\psi:\R^L \to \R$ is $\lambda$-\emph{smooth} w.r.t.\ $\|\cdot\|$ if for every $h_0,h\in \R^L$, letting $g_0$ be a subgradient of $\psi$ at $h_0$, we have
\begin{equation}
\label{eq:smooth}
\psi(h) \le \psi(h_0) + \langle g_0, h - h_0 \rangle + \frac \lambda 2 \|h - h_0\|^2.
\end{equation}
\end{definition}

If a convex function $\psi:\R^L\to \R$ is $\lambda$-smooth, it is necessarily differentiable with
the \emph{subgradient} $g_0$ being the \emph{gradient} at $h_0$, implying the uniqueness of the subgradient. We can thus uniquely define $\nabla \psi(h_0) := g_0$.

We additionally assume that $\varphi$ is a closed function on the bounded domain $\cG$ (and recall that $\psi$ is the convex conjugate of $\varphi$). This ensures that $\nabla \psi(h)\in \cG$ for every $h\in \R^L$ because the maximizer $g\in \cG$ in \eqref{eq:conjugate-1} achieves $\Gamma(g,h) = 0$, and thus $\nabla \psi(h) = g\in \cG$.

The Mirror Descent algorithm runs in iterations $k = 1,\ldots,K$, where each iteration $k$ updates a conjugate pair $(g_{k},h_{k})\in \cG\times \R^L$ to a new conjugate pair $(g_{k+1},h_{k+1})$ given a signal vector $f\in \R^L$ as follows:
\begin{align}
h_{k+1} & \gets h_{k} - f;\notag\\
g_{k+1} & \gets \nabla \psi(h_{k + 1}).\label{eq:mirror-descent}
\end{align}

The following standard lemma is central to analyzing mirror descent. It allows us to track the progress of the update from $(g_k,h_k)$ to $(g_{k + 1},h_{k + 1})$ using the nonnegative quantity $\Gamma(g^*,h_k)$ as a potential function for an arbitrary $g^*\in \cG$. The lemma shows that the potential function $\Gamma(g^*,h_k)$ must significantly decrease if $\langle g_{k} - g^*, f\rangle$ is large relative to $\|f\|^2$.

\begin{lemma}[Mirror Descent]
\label{lm:mirror-descent}
Let $\varphi:\cG\to \R$ be a closed convex function on a bounded convex set $\cG\subseteq \R^L$, and let $\psi:\R^L\to \R$ be its convex conjugate.
Assume that $\psi$ is $\lambda$-smooth w.r.t.\ norm $\|\cdot \|$ for some $\lambda \ge 0$.
Assume conjugate pairs $(g_k,h_k), (g_{k + 1}, h_{k + 1})\in \cG\times\R^L$ satisfy \eqref{eq:mirror-descent} for some $f\in \R^L$.
Then for every $g^*\in \cG$,
\begin{equation}
\label{eq:mirror}
\Gamma(g^*,h_k) - \Gamma(g^*,h_{k + 1}) \ge \langle g_{k} - g^*, f\rangle  - \frac{\lambda}{2} \|f\|^2.
\end{equation}
\end{lemma}
\begin{proof}
We start by simplifying \eqref{eq:mirror} using the following equations:
\begin{align*}
f & = h_k - h_{k + 1},\\
\Gamma (g^*,h_k) & = \varphi(g^*) + \psi(h_k) - \langle g^*, h_k \rangle,\\
\Gamma (g^*,h_{k + 1}) & = \varphi(g^*) + \psi(h_{k + 1}) - \langle g^*, h_{k + 1} \rangle.
\end{align*}
Our goal \eqref{eq:mirror} simplifies to
\[
\psi(h_k) - \psi(h_{k + 1}) \ge \langle g_k, h_k - h_{k + 1}\rangle  - \frac \lambda 2 \|h_k - h_{k + 1}\|^2.
\]
Since $(g_k, h_k)$ is a conjugate pair, we have $g_k = \nabla \psi(h_k)$. The inequality above then follows immediately from \eqref{eq:smooth} by setting $(g_0,h_0) = (g_k,h_k), h = h_{k + 1}$.
\end{proof}

\subsection{The Multiplicative Weights Algorithm}
\label{sec:mw}
The Multiplicative Weights algorithm is a special case of Mirror Descent when we choose $\cG = \Delta_L$ and choose $\varphi$ to be the negative Shannon entropy
\begin{equation}
\label{eq:mw-varphi}
\varphi(g) := \sum_{y = 1}^L g_y \ln g_y \le 0 \quad \text{for every $g= (g_1,\ldots,g_L)\in \Delta_L$}.
\end{equation}
Its convex conjugate $\psi$ is the following function:
\begin{equation}
\label{eq:mw-psi}
\psi(h) = \ln\left(\sum_{y = 1}^L e^{h_y}\right) \quad \text{for every $h = (h_1,\ldots,h_L)\in \R^L$}.
\end{equation}
The gradient of $\psi$ is the softmax function:
\[
\nabla \psi(h) = \softmax(h):= \frac{1}{\sum_{y = 1}^L e^{h_y}}(e^{h_1},\ldots,e^{h_L})\in \Delta_L.
\]
Each $e^{h_y}$ is often interpreted as the (unnormalized) \emph{weight} on label $y\in [L]$, and the softmax function normalizes these weights to a probability distribution over $[L]$. Thus the additive update on $h_y$ in \eqref{eq:mirror-descent} becomes a multiplicative update on the weight $e^{h_y}$, hence the name multiplicative weights.
\begin{lemma}
\label{lm:mw-smooth}
The function $\psi$ in \eqref{eq:mw-psi} is $1$-smooth in $\|\cdot\|_\infty$.
Moreover, the softmax function is a $1$-Lipschitz multi-variate function from $\|\cdot \|_\infty$ to $\|\cdot \|_1$. That is, for every $h,\hat h\in \R^L$,
\[
\|\softmax(h) - \softmax(\hat h)\|_1 \le \|h - \hat h\|_\infty.
\]
\end{lemma}
We give a proof of \Cref{lm:mw-smooth} using Pinsker's inequality in \Cref{sec:proof-mw-smooth}.
\section{Enhanced Regularity/Leakage-Simulation Lemma}
\label{sec:enhanced-regularity}

As we discuss in \Cref{sec:proof-idea}, the key to proving our main theorems (\Cref{thm:intro-main-1,thm:intro-unif}) is to establish an enhanced regularity/leakage-simulation lemma that allows us to construct a low-complexity function $s:\zo^n\to \Delta_L$ satisfying a unified distinguishing condition that captures both multiaccuracy \eqref{eq:proof-idea-3} and weight-restricted calibration \eqref{eq:proof-idea-2}. We formally state and prove this result in both the nonuniform and uniform settings.
\subsection{Nonuniform Setting}
We consider general distinguishers $\tau:\zo^n\times\Delta_L\to [-1,1]^L$ that take $x\in \zo^n$ as input together with a corresponding label distribution $v\in \Delta_L$.
Given a distinguisher $\tau:\zo^n\times\Delta_L\to [-1,1]^L$ and a function $g:\zo^n\to \Delta_L$, we write $\tau_g:\zo^n\to [-1,1]^L$ for the function $\tau_g(x):=\tau(x,g(x))$.
For a family $\gendistfamily$ of distinguishers, a $\gendistfamily$-\emph{oracle circuit} is a circuit that, in addition to ordinary gates, may use oracle gates labelled by functions $\tau\in\gendistfamily$.
We measure such a circuit by its ordinary circuit size and by the number of $\gendistfamily$-oracle gates.\footnote{We remark that in contrast to the Complexity Theoretic-Regularity Lemma of \cite{ttv}, the simulator $s$ in our Enhanced Regularity Lemma makes {\em adaptive} calls to the $\gendistfamily$ oracles.  In fact, in \cite{ttv} $s$ is of the form $s(x)=C(\tau_1(x),\ldots,\tau_k(x))$, where $\tau_1,\ldots,\tau_k\in \gendistfamily$, $k$ is the number of oracle gates, and $C$ is a small circuit. (Recall that in the ordinary Complexity-Theoretic Regularity Lemma, the distinguishers have the same domain as $g^*$ and $s$.)  This distinction may make the Enhanced Lemma difficult to use in contexts where it is important to preserve parallel complexity, e.g. when want $s$ to be computed by a small constant-depth threshold circuit ($\mathrm{TC}^0$) when the distinguishers are.}

\begin{theorem}[Nonuniform Enhanced Regularity/Leakage-Simulation Lemma for General Distinguishers]
\label{thm:reg-non-unif}
    Let $\gendistfamily$ be a family of distinguishers $\tau:\zo^n\times \Delta_L\to [-1,1]^L$. For every distribution $\mu$ on $\zo^n$, every $g^*:\zo^n\to \Delta_L$, and every $\varepsilon\in (0,1/2)$, there exists a simulator $s:\zo^n\to \Delta_L$ such that
    \begin{enumerate}
        \item $s$ has a $\gendistfamily$-oracle circuit of ordinary circuit size
        \[
        L\cdot \poly(1/\varepsilon,\log L)
        \]
        and uses $O((\log L)/\varepsilon^2)$ $\gendistfamily$-oracle gates;
        \item $s$ is indistinguishable from $g^*$ w.r.t.\ $\gendistfamily$:
        \[
        \langle s - g^*,\tau_s\rangle \le \varepsilon\quad \text{for every $\tau\in \gendistfamily$}.
        \]
    \end{enumerate}
    In particular, if every $\tau\in\gendistfamily$ has circuit complexity at most $T$, then $s$ has ordinary circuit complexity
    \[
    O\left(\frac{T\log L}{\varepsilon^2} + L\cdot \poly(1/\varepsilon,\log L)\right).
    \]
\end{theorem}

We prove \Cref{thm:reg-non-unif} by constructing $s$ using \Cref{alg:1} based on the Multiplicative Weights technique discussed in \Cref{sec:mw}.
In particular, our analysis tracks the Fenchel-Young divergence $\Gamma = \Gamma_{\varphi,\psi}$ defined by the negative Shannon entropy $\varphi$ in \eqref{eq:mw-varphi} and its convex conjugate $\psi$ in \eqref{eq:mw-psi}. For every $g:\zo^n\to \Delta_L$ and $h:\zo^n\to \R^L$, we define
\[
\Gamma(g,h):=\E_{x\sim\mu}[\Gamma(g(x),h(x))].
\]
For any norm $\|\cdot \|$ on $\R^L$ and any function $h:\zo^n\to \R^L$, we define
\[
\|h\|:= \max_{x\in \zo^n}\|h(x)\|.
\]

As we will show, \Cref{thm:reg-non-unif} follows from \Cref{thm:terminate,lm:circuit-comp} below.

\SetKw{KwGiven}{Given:}
\begin{algorithm}
\caption{Nonuniform Enhanced Regularity/Leakage-Simulation Lemma for General Distinguishers}
\label{alg:1}
\KwGiven{A family $\gendistfamily$ of functions $\tau:\zo^n\times\Delta_L\to [-1,1]^L$; a function $g^*:\zo^n \to \Delta_L$; parameter $\varepsilon\in (0,1)$; distribution $\mu$ on $\zo^n$;}\\
Initialize $h_0:\zo^n\to \R^L$ as the constant function where $h_0(x) = (0,\ldots,0)$ for every $x\in \zo^n$\;
Compute $\hat g_0:\zo^n\to \Delta_L$ s.t.\ $\|\hat g_0 - g_0\|_1 \le \varepsilon/10$, where $g_0:= \softmax \circ h_0$ \label{line:update-g-0}\;
$k\gets 0$\;
$\mathtt{updated}\gets \mathtt{TRUE}$\;
\While{$\mathtt{updated}$}{
    $\mathtt{updated}\gets \mathtt{FALSE}$\;
    \If{there exists $\tau\in \gendistfamily$ such that $\langle \hat g_k - g^*,\tau_{\hat g_k}\rangle > \varepsilon$ \label{line:if-tau}}{
        $h_{k + 1} \gets h_k - \varepsilon \tau_{\hat g_k}$ \label{line:update-tau}\;
        Compute $\hat g_{k + 1}:\zo^n\to \Delta_L$ s.t.\ $\|\hat g_{k + 1} - g_{k + 1}\|_1 \le \varepsilon/10$, where $g_{k + 1}:= \softmax \circ h_{k + 1}$ \label{line:update-g-tau}\;
        $k\gets k + 1$\;
        $\mathtt{updated}\gets \mathtt{TRUE}$\;
    }
}
\Return $s:=\hat g_k$\;
\end{algorithm}

\begin{lemma}
\label{thm:terminate}
    \Cref{alg:1} terminates after $k = O((\log L)/\varepsilon^2)$ updates. Moreover, the output $s$ produced by the algorithm satisfies
    \begin{equation}
    \label{eq:boosting-goal}
        \langle s - g^*,\tau_s \rangle \le \varepsilon \quad \text{for every $\tau\in \gendistfamily$}.
    \end{equation}
\end{lemma}
\begin{proof}
As we will show, throughout the algorithm, for every $k$, we have
\begin{equation}
\label{eq:potential-decrease}
\Gamma(g^*,h_k) < \Gamma(g^*,h_0) - k \cdot \varepsilon^2/4.
\end{equation}
The definition of the Fenchel-Young divergence $\Gamma$ ensures that  $\Gamma(g^*,h_k) \ge 0$.
We also have
\begin{align*}
\Gamma(g^*,h_0) = \E_{x\sim\mu}\Gamma(g^*(x),h_0(x)) & = \E_{x\sim\mu}[\varphi(g^*(x)) + \psi(h_0(x)) - \langle g^*(x), h_0(x) \rangle]\\
& = \E_{x\sim\mu}[\varphi(g^*(x)) + \psi(h_0(x))]\\
& \le \E_{x\sim\mu}[\psi(h_0(x))]\\
& = \ln L. \tag{by the definitions of $\varphi$ and $\psi$ in \eqref{eq:mw-varphi} and \eqref{eq:mw-psi}}
\end{align*}
Therefore, \eqref{eq:potential-decrease} implies that $k = O((\log L)/\varepsilon^2)$. Moreover, when \Cref{alg:1} terminates, the condition at \Cref{line:if-tau} is violated, implying \eqref{eq:boosting-goal}.

It remains to prove \eqref{eq:potential-decrease}.
We will prove the stronger claim that whenever we update from $h_k$ to $h_{k + 1}$ at \Cref{line:update-tau},
\begin{equation}
\label{eq:potential-decrease-1}
\Gamma(g^*,h_{k + 1}) < \Gamma(g^*,h_k) - \varepsilon^2/4.
\end{equation}
Let $\tau\in \gendistfamily$ be the distinguisher used in the update. By \Cref{lm:mirror-descent}, \Cref{lm:mw-smooth} and the fact that $\|\tau_{\hat g_k}\|_\infty \le 1$, we have
\begin{equation}
\label{eq:potential-decrease-2}
\langle g_k - g^*, \varepsilon \tau_{\hat g_k}\rangle \le \Gamma(g^*,h_k) - \Gamma(g^*,h_{k + 1}) + \frac 1 2\cdot \varepsilon^2.
\end{equation}
By the condition at \Cref{line:if-tau} and the fact that $\|\hat g_k - g_k\|_1\le \varepsilon/10$, we have
\begin{equation}
\label{eq:potential-decrease-3}
\langle g_k - g^*, \tau_{\hat g_k} \rangle > (9/10)\varepsilon.
\end{equation}
Plugging \eqref{eq:potential-decrease-3} into \eqref{eq:potential-decrease-2} proves \eqref{eq:potential-decrease-1}.
\end{proof}
\begin{lemma}[Controlling oracle circuit complexity]
\label{lm:circuit-comp}
    There exists an implementation of \Cref{alg:1} such that, for every $k$, the function $\hat g_k$ has a $\gendistfamily$-oracle circuit of ordinary circuit size $(k + 1)L \cdot \poly(1/\varepsilon, \log L)$ using at most $k$ $\gendistfamily$-oracle gates, and each coordinate of $\hat g_k(x)$ for every $x\in \zo^n$ is represented in binary format with $O(\log (L/\varepsilon))$ bits.
\end{lemma}
We will use the following lemma to prove \Cref{lm:circuit-comp}:
\begin{lemma}[Approximate computation of softmax]
\label{lm:approx-softmax}
    For $\varepsilon\in (0,1/2), B \ge 2$, let $q,\hat q\in [-B,B]^L$ be vectors such that $\|\hat q - q\|_\infty \le \varepsilon/3$.
    Assume that each coordinate of $\hat q$ is rational and is represented in binary format with finitely many bits.
    There exists a circuit of size $L \cdot \poly(B,\log(1/\varepsilon),\log L)$ that, given $\hat q$, computes $v\in \Delta_L$ such that $\|v - \softmax(q)\| \le \varepsilon$, where each coordinate of $v$ is represented in binary format with $O(\log(L/\varepsilon))$ bits.
\end{lemma}
\begin{proof}
Since $\hat q\in [-B,B]^L$, the integer part of each coordinate of $\hat q$ has $O(\log B)$ bits. We keep only the first $O(\log (1/\varepsilon) )$ bits in the fractional part of each coordinate of $\hat q$. We can make sure that this truncated $\hat q$ still satisfies $\|\hat q - q\|_\infty \le \varepsilon/2$.
Suppose $\hat q = (\hat q_1,\ldots,\hat q_L)$.
For each $y = 1,\ldots,L$,
    using Taylor's expansion of $e^x$, we can compute $p_y$ in $\poly(B,\log(1/\varepsilon))$ circuit complexity such that
    \[
    p_y / e^{\hat q_y} \in [1 - \varepsilon/10, 1 + \varepsilon/10].
    \]
This implies that $p_y = e^{\tilde q_y}$ for some $\tilde q_y\in [\hat q_y - \varepsilon/4, \hat q_y + \varepsilon/4]\subseteq [q_y - 3\varepsilon/4, q_y + 3\varepsilon/4]$. That is,
\begin{equation}
\label{eq:q-tilde}
\|\tilde q - q\|_\infty \le 3\varepsilon/4.
\end{equation}
Define
\[
v := \frac{1}{\sum_{y\in [L]}p_y}(p_1,\ldots,p_L) = \softmax(\tilde q).
\]
Since each $p_y$ can be computed in circuit complexity $\poly(B,\log(1/\varepsilon))$, we have that $v$ can be computed in circuit complexity $L \cdot \poly(B,\log(1/\varepsilon),\log L)$.
By \eqref{eq:q-tilde} and the Lipschitzness of $\softmax$ (\Cref{lm:mw-smooth}), we have
\[
\|v - \softmax(q)\|_1 = \|\softmax(\tilde q) - \softmax (q)\|_1 \le 3\varepsilon/4.
\]
The proof is completed by noting that we can round each coordinate of $v$ to the first $O(\log(L/\varepsilon))$ bits while ensuring that the updated $v$ satisfies $v\in \Delta_L$ and $\|v - \softmax(q)\|_1 \le \varepsilon.$
\end{proof}

We are now ready to prove \Cref{lm:circuit-comp}.

\begin{proof}[Proof of \Cref{lm:circuit-comp}]
Let $\kappa = O((\log L)/\varepsilon^2)$ be an upper bound on the total number of updates in \Cref{alg:1}, as guaranteed by \Cref{thm:terminate}.
Throughout the algorithm, we maintain a circuit $C_k$ that computes $x\mapsto (\hat h_k(x),\hat g_k(x))$ where $\hat h_k:\zo^n\to \R^L,\hat g_k:\zo^n \to \Delta_L$ satisfy the following:
\begin{enumerate}
    \item $\|\hat h_k\|_\infty \le k\varepsilon \le \kappa \varepsilon$, so the integer part of each coordinate of $\hat h_k$ has at most $O(\log(\kappa\varepsilon))$ bits;
    \item $\|\hat h_k - h_k\|_\infty \le k\cdot (\varepsilon/30 \kappa) \le \varepsilon/30$;
    \item $\|\hat g_k - g_k\|_1 \le \varepsilon/10$ as required by \Cref{line:update-g-0,line:update-g-tau};
    \item Each coordinate of $\hat h_k(x)$ for every $x\in \zo^n$ is represented in binary format with $O(\log(\kappa/\varepsilon))$ bits after the decimal point;
    \item Each coordinate of $\hat g_k(x)$ for every $x\in \zo^n$ is represented in binary format with $O(\log(L/\varepsilon))$ bits.
\end{enumerate}

We show that the ordinary circuit size of $C_k$ need not exceed $(k + 1)L \cdot \poly(1/\varepsilon, \log L)$ while using at most $k$ $\gendistfamily$-oracle gates.
It suffices to inductively prove an ordinary circuit size upper bound $T_k$ for $C_k$ such that
\begin{align*}
T_0 & = L \cdot \poly(1/\varepsilon, \log L),\\
T_{k + 1} & \le T_k + L \cdot \poly(1/\varepsilon, \log L),
\end{align*}
with one additional $\gendistfamily$-oracle gate at each update.

We prove this by induction on $k$. When $k = 0$, the algorithm initializes $h_0$ to be the constant zero function, so $\hat h_0$ can be computed in circuit complexity $O(L)$. By \Cref{lm:approx-softmax}, $\hat g_0$ can be computed in circuit complexity $L \cdot \poly(\log(1/\varepsilon),\log L)$.

Now suppose we have a $\gendistfamily$-oracle circuit $C_k$ with ordinary size at most $T_k$ and at most $k$ $\gendistfamily$-oracle gates that computes $x\mapsto (\hat h_k(x),\hat g_k(x))$ satisfying the five conditions above, and let us construct the circuit $C_{k + 1}$ that computes $x\mapsto (\hat h_{k + 1}(x), \hat g_{k + 1}(x))$. In \Cref{alg:1}, $h_{k + 1}$ is updated from $h_k$ at \Cref{line:update-tau} using $\tau_{\hat g_k}$ for some $\tau\in\gendistfamily$. We first compute $\hat \tau$ with $\|\hat \tau\|_\infty \le \varepsilon$ such that $\|\hat \tau - \varepsilon \tau_{\hat g_k}\|_\infty \le \varepsilon / 30 \kappa$, where each coordinate of $\hat \tau(x)$ need only have $O(\log(\kappa/\varepsilon))$ bits in binary representation. This can be done by first computing $\hat g_k(x)$ using $C_k$ and then using one $\gendistfamily$-oracle gate to evaluate $\tau(x,\hat g_k(x))$, followed by $L \cdot \poly(\log(\kappa/\varepsilon))$ ordinary gates for truncation and scaling. Then we subtract $\hat \tau$ from $\hat h_k$ to obtain $\hat h_{k + 1}$. The subtraction takes circuit complexity $L\cdot \poly(\log(\kappa/\varepsilon))$ and  ensures that $h_{k + 1}$ satisfies conditions 1,2, and 4.  Finally, by \Cref{lm:approx-softmax}, we can compute $\hat g_{k + 1}$ satisfying conditions 3 and 5 from $\hat h_{k + 1}$ in circuit complexity $O(L \cdot \poly(\kappa\varepsilon,\log(1/\varepsilon), \log L))$. The total ordinary circuit size of $C_{k + 1}$ is
\begin{align*}
& T_k + L\cdot \poly(\log(\kappa/\varepsilon)) + O(L \cdot \poly(\kappa \varepsilon,\log(1/\varepsilon), \log L)) \\
\le {} & T_k + L\cdot \poly(1/\varepsilon,\log L),
\end{align*}
where we used the fact that $\kappa = O((\log L)/\varepsilon^2)$. The number of $\gendistfamily$-oracle gates increases by one, as required.
\end{proof}

We now complete the proof of \Cref{thm:reg-non-unif} using \Cref{thm:terminate} and \Cref{lm:circuit-comp}.
\begin{proof}[Proof of \Cref{thm:reg-non-unif}]
\Cref{thm:terminate} ensures that the output function $s$ satisfies Condition 2 in \Cref{thm:reg-non-unif}. By \Cref{lm:circuit-comp}, we can make sure that $s$ has ordinary circuit size at most
\[
(k + 1)L\cdot \poly(1/\varepsilon,\log L) = L\cdot \poly(1/\varepsilon,\log L),
\]
and at most $k=O((\log L)/\varepsilon^2)$ $\gendistfamily$-oracle gates by \Cref{thm:terminate}. This establishes Condition 1 required by \Cref{thm:reg-non-unif}. If every $\tau\in\gendistfamily$ has circuit complexity at most $T$, replacing each oracle gate by such a circuit gives the claimed ordinary circuit complexity bound.
\end{proof}

\subsection{Uniform Setting}
To formally state and prove our enhanced regularity/leakage-simulation lemma in the uniform setting, we extend the unified distinguishing condition from \Cref{thm:reg-non-unif} to the uniform setting. To that end, we make the following definitions.
\begin{definition}[Induced Distribution]
    Let $\mu$ be a distribution over $\zo^n$. For a function $g:\zo^n\to \Delta_L$, we define $\mu_g$ as the distribution of $(x,y)\in \zo^n\times [L]$ where we first draw $x$ from $\mu$ and then draw $y$ from the distribution $g(x)$. For a pair of functions $g,g^*:\zo^n\to \Delta_L$, we define $\mu_{g,g^*}$ as the distribution of $(x,y,y^*)\in \zo^n\times [L] \times [L]$, where we first draw $x$ from $\mu$, and then independently draw $y\sim g(x),y^*\sim g^*(x)$.
\end{definition}
\begin{definition}[General Distinguishing Oracle]
\label{def:oracle-dist}
    An $(m,T)$-\emph{general distinguishing oracle} is a time-$T$ algorithm $\cA$ that takes $m$ data points $(x_1,v_1,y_1,y_1^*),\ldots,(x_m,v_m,y_m,y_m^*)\in \zo^n\times \Delta_L\times [L]\times [L]$ as input, and outputs (a succinct description of) a function $\tau:\zo^n\times\Delta_L\to [-1,1]^L$ such that given any $(x,v)\in \zo^n\times\Delta_L$, the function value $\tau(x,v)\in [-1,1]^L$ can be evaluated in time $T$.
\end{definition}

\begin{definition}[Distinguishability]
\label{def:unif-distinguishability}
Let $\mu$ be a distribution over $\zo^n$ and let $\cA$ be an $(m,T)$-general distinguishing oracle. For a pair of functions $g,g^*:\zo^n\to \Delta_L$, we say $g$ is $(\cA,\varepsilon,\delta)$-\emph{distinguishable} from $g^*$ if with probability at least $1-\delta$ over the random draw of $m$ i.i.d.\ data points $(x_1,y_1,y^*_1),\ldots,\allowbreak (x_m,y_m,y^*_m) \sim \mu_{g,g^*}$, when $\cA$ is given $(x_1,g(x_1),y_1,y_1^*),\ldots,(x_m,g(x_m),y_m,y_m^*)$ as input, the output function $\tau:\zo^n\times\Delta_L\to [-1,1]^L$ satisfies
    \[
    \langle g - g^*,\tau_g\rangle > \varepsilon.
    \]
Correspondingly, we say $g$ is $(\cA,\varepsilon,\delta)$-\emph{indistinguishable} from $g^*$ if it is not $(\cA,\varepsilon,\delta)$-distinguishable.
\end{definition}

In \Cref{def:unif-distinguishability} above, the success probability of an oracle can be amplified by independent repetition. Suppose we have an oracle with success probability $\alpha > 0$. We can construct an oracle with success probability at least $1-\delta$ by independently running the oracle $u=O(\alpha^{-1}\log(1/\delta))$ times to ensure that with probability at least $1-\delta/2$, at least one of the runs is successful. We can then perform a simple test using $O(\varepsilon^{-2}\log(u/\delta))$ additional fresh data points to select a successful run with probability at least $1 - \delta/2$ (with a constant-factor loss in $\varepsilon$). The overall success probability is at least $1-\delta$ by the union bound.

We are now ready to state the enhanced regularity/leakage-simulation lemma in the uniform setting:

\begin{theorem}[Uniform Enhanced Regularity/Leakage-Simulation Lemma for General Distinguishers]
\label{thm:reg-unif}
    Let $\cA$ be an $(m,T)$ general distinguishing oracle. Then there exists a simulating algorithm $\cS$ using $\cA$ as a subroutine such that for every $g^*:\zo^n\to\Delta_L$ the following holds:
    \begin{enumerate}
        \item Given $\varepsilon,\delta\in (0,1/2)$ as input, $\cS$ additionally takes $\poly(m,\log L, 1/\varepsilon,\log(1/\delta))$ i.i.d.\ data points from $\mu_{g^*}$ as input, and returns a function $s:\zo^n\to \Delta_L$ in time $\poly(m,T,L,1/\varepsilon, \log(1/\delta))$. Moreover, given an arbitrary $x\in \zo^n$, the function value $s(x)$ can also be evaluated in time $\poly(m,T,L,1/\varepsilon, \log(1/\delta))$.
        \item With probability at least $1-\delta$ over the random draw of the data points, the output function $s$ is $(\cA,\varepsilon, \Omega(\delta \varepsilon^2/\log L))$-indistinguishable from $g^*$.
    \end{enumerate}
\end{theorem}

\begin{algorithm}
\caption{Uniform Enhanced Regularity/Leakage-Simulation Lemma for General Distinguishers}
\label{alg:2}
\KwGiven{An $(m,T)$-general distinguishing oracle $\cA$ on domain $\zo^n$ and label space $[L]$; access to i.i.d.\ data points drawn from $\mu_{g^*}$ for distribution $\mu$ on $\zo^n$ and function $g^*:\zo^n\to \Delta_L$; parameters $\varepsilon\in (0,1), m'\in \Z_{>0}$;} \\
Initialize $h_0:\zo^n\to \R^L$ as the constant function where $h_0(x) = (0,\ldots,0)$ for every $x\in \zo^n$\;
Compute $\hat g_0:\zo^n\to \Delta_L$ s.t.\ $\|\hat g_0 - g_0\|_1 \le \varepsilon/20$, where $g_0:= \softmax \circ h_0$\;
$k\gets 0$\;
$\mathtt{updated}\gets \mathtt{TRUE}$\;
\While{$\mathtt{updated}$}{
    $\mathtt{updated}\gets \mathtt{FALSE}$\;
    Draw fresh i.i.d.\ data points $(x_1,y_1^*),\ldots,(x_m,y_m^*)$ from $\mu_{g^*}$\label{line:draw-1}\;
    Draw $y_i\sim \hat g_k(x_i)$ independently for $i = 1,\ldots,m$ \label{line:draw-1-1}\;
    Invoke oracle $\cA$ on data points $(x_i,\hat g_k(x_i),y_i,y_i^*)$ and obtain its output $\tau:\zo^n\times\Delta_L\to [-1,1]^L$ \label{line:invoke-tau}\;
    Draw fresh i.i.d.\ data points $(x_1,y^*_1),\ldots,(x_{m'},y^*_{m'})$ from $\mu_{g^*}$ \label{line:draw-2}\;
    \If{$\frac 1{m'} \sum_{i=1}^{m'} \langle \hat g_k(x_i) - \be_{y_i^*}, \tau(x_i,\hat g_k(x_i))\rangle > 3\varepsilon/4$ \label{line:unif-if-tau}}{
        $h_{k + 1} \gets h_k - (\varepsilon/2) \tau_{\hat g_k}$\;
        Compute $\hat g_{k+1}:\zo^n\to \Delta_L$ s.t.\ $\|\hat g_{k + 1} - g_{k + 1}\|_1 \le \varepsilon/20$, where $g_{k + 1}:= \softmax \circ h_{k + 1}$\;
        $k\gets k + 1$\;
        $\mathtt{updated}\gets \mathtt{TRUE}$\;
    }
}
\Return $s:=\hat g_k$\;
\end{algorithm}

We prove \Cref{thm:reg-unif} using \Cref{alg:2} as the simulating algorithm $\cS$ based on the same multiplicative weights idea as \Cref{alg:1}.
The difference is that instead of directly searching over a nonuniform family of distinguishers as in \Cref{line:if-tau} of \Cref{alg:1}, we obtain the relevant distinguisher efficiently by invoking the oracle $\cA$ at \Cref{line:invoke-tau} in \Cref{alg:2}.
Recall that for $y\in [L]$, we use $\be_y$ to denote the standard basis vector whose $y$-th coordinate is one.
This ensures that for every $v\in \Delta_L$, we have $\E_{y\sim v}[\be_y] = v$.

\Cref{thm:reg-unif} follows immediately from the following lemma about \Cref{alg:2}:

\begin{lemma}
\label{lm:terminate-unif}
For every $\varepsilon\in (0,1/2)$, there exist
\[
\kappa = O((\log L)/\varepsilon^2),\quad
\delta' = \Omega(\delta/\kappa),\quad
m' = O((\log(1/\delta'))/\varepsilon^2)
\]
such that when we run \Cref{alg:2} with parameters $(\varepsilon,m')$, with probability at least $1-\delta$, the following holds. \Cref{alg:2} terminates after $k \le \kappa$ updates and returns a function $s:\zo^n\to \Delta_L$ that is $(\cA,\varepsilon,\delta')$-indistinguishable from $g^*$. Also, \Cref{alg:2} runs in time $\poly(m,T,L,1/\varepsilon,\log(1/\delta))$, and for every $x\in \zo^n$, the function value $s(x)$ can be evaluated in time $\poly(m,T,L,1/\varepsilon,\log(1/\delta))$.
\end{lemma}

\begin{proof}
By the Chernoff bound, choosing $m' = O((\log(1/\delta'))/\varepsilon^2)$ is sufficient to ensure that each time \Cref{line:draw-2} is performed, with probability at least $1-\delta'$ we have
\begin{equation}
\label{eq:terminate-unif-1}
\left|\frac {1}{m'}\sum_{i = 1}^{m'}\langle \hat g_k(x_i) - \be_{y_i^*}, \tau(x_i,\hat g_k(x_i))\rangle - \langle \hat g_k - g^*,\tau_{\hat g_k} \rangle\right| < \varepsilon/4.
\end{equation}
By the definition of $(\cA,\varepsilon,\delta')$-distinguishability (\Cref{def:unif-distinguishability}), each time when \Cref{line:invoke-tau} is performed with $\hat g_k$ being $(\cA,\varepsilon,\delta')$-distinguishable, with probability at least $1-\delta'$ we have
\begin{equation}
\label{eq:terminate-unif-2}
\langle \hat g_k - g^*,\tau_{\hat g_k}\rangle > \varepsilon.
\end{equation}
By the union bound, we can choose $\delta' = \Omega(\delta/\kappa)$ to ensure that with probability at least $1-\delta$, all of the above are satisfied \emph{simultaneously} in the first $\kappa + 1$ iterations of the while loop. It thus remains to show that when this occurs, the following two conditions required by \Cref{lm:terminate-unif} are both satisfied:
\begin{enumerate}
    \item \label{item:terminate-unif-1} \Cref{alg:2} terminates after $k \le \kappa = O((\log L)/\varepsilon^2)$ updates and returns a function $s:\zo^n\to \Delta_L$ that is $(\cA,\varepsilon,\delta')$-indistinguishable from $g^*$.
    \item \label{item:terminate-unif-2} \Cref{alg:2} runs in time $\poly(m,T,L,1/\varepsilon,\log(1/\delta))$, and for every $x\in \zo^n$, the function value $s(x)$ can be evaluated in time $\poly(m,T,L,1/\varepsilon,\log(1/\delta))$.
\end{enumerate}
To prove that the output function $s = \hat g_k$ is $(\cA,\varepsilon,\delta')$-indistinguishable from $g^*$ as required by \Cref{item:terminate-unif-1}, it suffices to show that \eqref{eq:terminate-unif-2} is violated. By \eqref{eq:terminate-unif-1}, it suffices to show that the condition in the ``if'' clause at \Cref{line:unif-if-tau} is violated. This is indeed true for the final output $s = \hat g_k$ because otherwise \Cref{alg:2} would proceed to the next iteration of the ``while'' loop. Similarly, by \eqref{eq:terminate-unif-1}, if the condition at \Cref{line:unif-if-tau} is satisfied, we have
\[
\langle \hat g_k - g^*,\tau_{\hat g_k} \rangle > \varepsilon/2.
\]
This allows us to prove that \Cref{alg:2} terminates after $k\le \kappa = O((\log L)/\varepsilon^2)$ updates following the same idea in the proof of \Cref{thm:terminate}, completing the proof of \Cref{item:terminate-unif-1}.
Also, \Cref{item:terminate-unif-2} can be proved similarly to \Cref{lm:circuit-comp} using the fact that $\cA$ runs in time $T$ and its output function $\tau$ can also be evaluated in time $T$, according to \Cref{def:oracle-dist}. We omit the details.
\end{proof}

\section{Generalized and Unified Pseudoentropy Characterizations}
In this section, we complete the proofs of our main results \Cref{thm:intro-main-1,thm:intro-unif,thm:intro-unif-converse}.
\subsection{Nonuniform Setting}
\label{sec:proof-main-non-unif}
\begin{proof}[Proof of \Cref{thm:intro-main-1}]
For $\varphi \in \Phi$, by our assumption that every $\nabla \varphi$  can be efficiently approximated within $\ell_\infty$ error $\varepsilon/4$, there exists a function $r_\varphi:\Delta_L\to [-1,1]^L$ of circuit complexity at most $T_\Phi$ such that
\begin{equation}
\label{eq:rphi}
\|r_\varphi(v) - \nabla \varphi(v)\|_\infty \le \varepsilon /4 \quad \text{for every $v\in \Delta_L$}.
\end{equation}
    By \eqref{eq:gap-calibration}, the guarantee \eqref{eq:intro-main-1} is equivalent to
    \[
    \langle s - g^*, \nabla \varphi\circ s \rangle \le \varepsilon \quad \text{for every $\varphi\in \Phi$}.
    \]
   By \eqref{eq:rphi}, a sufficient condition for the guarantee \eqref{eq:intro-main-1} is
    \[
    \langle s - g^*,  r_\varphi \circ s \rangle \le \varepsilon /2 \quad \text{for every $\varphi\in \Phi$}.
    \]
    Now the first half of \Cref{thm:intro-main-1} (up to Equation~\eqref{eq:intro-main-1}) follows directly from \Cref{thm:reg-non-unif}, where we let $\gendistfamily$ contain the functions $\tau_f^+(x,v):=f(x)$ and $\tau_f^-(x,v):=-f(x)$ for every $f:\zo^n\to [-1,1]^L$ with circuit complexity at most $T$, and the functions $\tau_\varphi(x,v):=r_\varphi(v)$ for every $\varphi\in\Phi$. Every function in $\gendistfamily$ has circuit complexity at most $T+T_\Phi$, and we replace $\varepsilon$ with $\varepsilon/2$.

    To prove the second half of \Cref{thm:intro-main-1} (i.e.\ the reverse direction), we note the following identity for every $g,g^*,s:\zo^n\to \Delta_L$:
    \[
\D_\varphi(g^*\|g) = \D_\varphi(s\|g) + \langle s - g^*,\nabla \varphi \circ g \rangle + (\H_\varphi(s) - \H_\varphi(g^*)).
    \]
    This identity follows immediately from the definitions of $\H_\varphi$ and $\D_\varphi$ in \Cref{def:entropy,def:bregman}.
    By \eqref{eq:rphi} and the fact that $\D_\varphi(s\|g) \ge 0$, we have
    \begin{align}
    \D_\varphi(g^*\|g) & \ge \langle s - g^*,\nabla \varphi \circ g \rangle + (\H_\varphi(s) - \H_\varphi(g^*))\notag \\
    & \ge \langle s - g^*, r_\varphi \circ g \rangle + (\H_\varphi(s) - \H_\varphi(g^*)) - \varepsilon/2.\label{eq:converse-1}
    \end{align}
Since $r_\varphi$ has circuit complexity at most $T_\Phi$, when $g$ has circuit complexity at most $T$, we have that $r_\varphi\circ g$ has circuit complexity at most $O(T + T_\Phi)$, so for $s\in \ind(g^*, O(T + T_\Phi),\varepsilon/2)$, we have
\[
\langle s - g^*, r_\varphi \circ g \rangle \ge -\varepsilon/2.
\]
Plugging this into \eqref{eq:converse-1} proves \eqref{eq:converse}.
\end{proof}

\subsection{Uniform Setting}
\label{sec:proof-main-unif}
We state and prove \Cref{thm:unif}, which is the formal version of \Cref{thm:intro-unif}. After that we prove \Cref{thm:intro-unif-converse}.
We need the following definitions to state \Cref{thm:unif}.
\begin{definition}[Distinguishing and Calibration Oracles]
\label{def:dist-calib-oracles}
Let $\mu$ be a fixed distribution on $\zo^n$.
An $(m,T)$-\emph{distinguishing oracle} is a time-$T$ algorithm $\cA$ that takes data points
\[
(x_i,y_i,y_i^*)\in \zo^n\times [L]\times [L]\quad (i=1,\ldots,m)
\]
as input, and outputs (a succinct description of) a function $f:\zo^n\to [-1,1]^L$ such that given any $x\in \zo^n$, the function value $f(x)\in [-1,1]^L$ can be evaluated in time $T$.
For such an oracle $\cA$ and functions $g,g^*:\zo^n\to \Delta_L$, we say $g$ is $(\cA,\varepsilon,\delta)$-\emph{distinguishable} from $g^*$ if with probability at least $1-\delta$ over the random draw of $((x_i,y_i,y_i^*))_{i=1}^m\sim \mu_{g,g^*}^{\otimes m}$, when $\cA$ is given these data points as input, the output function $f:\zo^n\to [-1,1]^L$ satisfies
\[
\langle g - g^*, f\rangle > \varepsilon.
\]
Correspondingly, we say $g$ is $(\cA,\varepsilon,\delta)$-\emph{indistinguishable} from $g^*$ if it is not $(\cA,\varepsilon,\delta)$-distinguishable.

An $(m,T)$-\emph{calibration oracle} is a time-$T$ algorithm $\cB$ that takes data points
\[
(v_i,y_i,y_i^*)\in \Delta_L\times [L]\times [L]\quad (i=1,\ldots,m)
\]
as input, and outputs (a succinct description of) a function $r:\Delta_L\to [-1,1]^L$ such that given any $v\in \Delta_L$, the function value $r(v)\in [-1,1]^L$ can be evaluated in time $T$.
\end{definition}

\begin{definition}[Weak Agnostic Calibration Oracle]
\label{def:weak-ag}
Let $\mu$ be a fixed distribution on $\zo^n$.
Let $\cR$ be a family of functions $r:\Delta_L\to [-1,1]^L$, and let $\varepsilon_1,\varepsilon,\delta\in (0,1)$ be parameters.
We say an $(m,T)$-calibration oracle $\cB$ is an $(\cR,\varepsilon_1,\varepsilon,\delta)$-\emph{weak agnostic calibration oracle} if the following holds. For every pair $g,g^*:\zo^n\to \Delta_L$, if there exists $r\in \cR$ such that
\[
\langle g - g^*, r\circ g\rangle> \varepsilon_1,
\]
then with probability at least $1-\delta$ over the random draw of $((x_i,y_i,y_i^*))_{i=1}^m\sim \mu_{g,g^*}^{\otimes m}$, when $\cB$ is given $((g(x_i),y_i,y_i^*))_{i=1}^m$ as input, the output function $r_{\mathsf{output}}:\Delta_L\to [-1,1]^L$ satisfies
\[
\langle g - g^*,r_{\mathsf{output}}\circ g\rangle> \varepsilon.
\]
\end{definition}

The above definition is analogous to the standard definition of weak agnostic learning \cite{weak-agnostic-1,weak-agnostic-2}. When $\cR$ is a finite class, we can implement an $(\cR, \varepsilon_1, \varepsilon_1/2,\delta)$-weak agnostic calibration oracle using $m = O(\varepsilon_1^{-2}\log(|\cR|/\delta))$ data points by performing \emph{empirical risk minimization (ERM)}, where we output $r\in \cR$ that maximizes the empirical correlation:
\[
\frac 1m \sum_{i = 1}^m\langle v_i - \be_{y_i^*},r(v_i)\rangle.
\]

\begin{theorem}[Formal statement of \Cref{thm:intro-unif}]
\label{thm:unif}
Consider domain $\zo^n$ and label space $[L]$. Let $\Phi$ be a family of convex functions $\varphi:\Delta_L\to \R$ with a fixed choice of subgradient satisfying $\nabla\varphi(v)\in [-1,1]^L$ for every $\varphi\in\Phi$ and $v\in\Delta_L$, and let
\[
\cR=\{\nabla\varphi:\Delta_L\to [-1,1]^L\mid \varphi\in\Phi\}.
\]
For every $\varepsilon,\delta\in (0,1/2)$, there exists $\delta' = \Omega(\delta \varepsilon^2/\log L)$ and a simulating algorithm $\cS$ such that the following holds for every $g^*:\zo^n\to\Delta_L$, distribution $\mu$ on $\zo^n$, and oracles $\cA,\cB$:

    \begin{enumerate}
        \item If $\cA$ is an $(m,T)$ distinguishing oracle and $\cB$ is an $(m,T)$ calibration oracle, then $\cS$ takes $\poly(m,\log L, 1/\varepsilon,\log(1/\delta))$ i.i.d.\ data points from $\mu_{g^*}$ as input, uses $\cA$ and $\cB$ as oracles, and returns a function $s:\zo^n\to \Delta_L$ in time $\poly(m,T,L,1/\varepsilon, \log(1/\delta))$. The value $s(x)$ for an arbitrary $x\in \zo^n$ can also be evaluated within the same time bound.
        \item If in addition $\cB$ is an $(\cR,\varepsilon_1,\varepsilon,\delta')$ weak agnostic calibration oracle for some $\varepsilon_1 > 0$, then with probability at least $1-\delta$ over the random draw of the input data points, the output function $s$ is $(\cA,\varepsilon, \delta')$-indistinguishable from $g^*$ and satisfies
        \begin{equation}
        \label{eq:unif}
        \H_\varphi(s) - \H_\varphi(g^*) \ge \D_\varphi(g^*\|s) - \varepsilon_1 \quad \text{for every $\varphi\in \Phi$.}
        \end{equation}
    \end{enumerate}
\end{theorem}
\begin{proof}
    By \eqref{eq:gap-calibration}, the guarantee \eqref{eq:unif} is equivalent to
    \[
    \langle s - g^*, \nabla \varphi\circ s\rangle \le \varepsilon_1 \quad \text{for every $\varphi\in \Phi$}.
    \]

    We construct a general distinguishing oracle $\cC$ from $\cA$ and $\cB$ and then apply \Cref{thm:reg-unif}. On input samples $(x_i,v_i,y_i,y_i^*)$ drawn from $\mu_{g,g^*}$ with $v_i=g(x_i)$, the oracle $\cC$ splits its samples into three independent blocks. It runs $\cA$ on the first block after discarding the $v_i$'s, obtaining $f:\zo^n\to[-1,1]^L$, and runs $\cB$ on the second block after discarding the $x_i$'s, obtaining $r:\Delta_L\to[-1,1]^L$. These two outputs define general distinguishers
    \[
    \tau_{\cA}(x,v):=f(x)
    \qquad\text{and}\qquad
    \tau_{\cB}(x,v):=r(v).
    \]
    The oracle $\cC$ uses the third block to estimate the two correlations
    \[
    \langle g-g^*,(\tau_{\cA})_g\rangle
    \qquad\text{and}\qquad
    \langle g-g^*,(\tau_{\cB})_g\rangle
    \]
    by their empirical analogues $\frac1q\sum_i\langle v_i-\be_{y_i^*},\tau_{\cA}(x_i,v_i)\rangle$ and $\frac1q\sum_i\langle v_i-\be_{y_i^*},\tau_{\cB}(x_i,v_i)\rangle$, and returns whichever of $\tau_{\cA}$ and $\tau_{\cB}$ has the larger empirical correlation. By taking $q=O(\varepsilon^{-2}\log(1/\delta'))$, the empirical estimates are within a constant multiple of $\varepsilon$ of the true correlations except with probability at most a constant multiple of $\delta'$.

    Thus, up to constant-factor losses in $\varepsilon$ and $\delta'$ that are absorbed in the $\Omega(\delta\varepsilon^2/\log L)$ notation, $\cC$ is an $(m',T')$-general distinguishing oracle with $m'=\poly(m,1/\varepsilon,\log(1/\delta'))$ and $T'=\poly(T,L,1/\varepsilon,\log(1/\delta'))$. Applying \Cref{thm:reg-unif} to $\cC$ gives an algorithm $\cS$ that uses $\cA$ and $\cB$ as subroutines and, with probability at least $1-\delta$, outputs a function $s$ that is $(\cC,\varepsilon/2,\delta')$-indistinguishable from $g^*$.

    We now translate this conclusion back to $\cA$ and $\cB$. If $s$ were $(\cA,\varepsilon,\delta')$-distinguishable from $g^*$, then the candidate $\tau_{\cA}(x,v)=f(x)$ produced inside $\cC$ would have correlation larger than $\varepsilon$ with $s-g^*$ with probability at least $1-\delta'$, and the validation step in $\cC$ would make $\cC$ output a general distinguisher with correlation larger than $\varepsilon/2$. This contradicts the $(\cC,\varepsilon/2,\delta')$-indistinguishability of $s$. Hence $s$ is $(\cA,\varepsilon,\delta')$-indistinguishable from $g^*$.

    Similarly, suppose there is some $\varphi\in\Phi$ such that
    \[
    \langle s - g^*, \nabla\varphi\circ s\rangle > \varepsilon_1.
    \]
    Since $\nabla\varphi\in\cR$ and $\cB$ is an $(\cR,\varepsilon_1,\varepsilon,\delta')$ weak agnostic calibration oracle, the candidate $\tau_{\cB}(x,v)=r(v)$ produced inside $\cC$ would, with probability at least $1-\delta'$, have correlation larger than $\varepsilon$ with $s-g^*$. The validation step would again make $\cC$ output a general distinguisher with correlation larger than $\varepsilon/2$, contradicting the $(\cC,\varepsilon/2,\delta')$-indistinguishability of $s$. Therefore
    \[
    \langle s - g^*, \nabla \varphi\circ s\rangle \le \varepsilon_1 \quad \text{for every $\varphi\in \Phi$},
    \]
    which is equivalent to \eqref{eq:unif}.
\end{proof}

\begin{proof}[Proof of \Cref{thm:intro-unif-converse}]
Let $r_\varphi:\Delta_L\to [-1,1]^L$ be the function computed using $\cB$, so that
\[
\|r_\varphi(v)-\nabla\varphi(v)\|_\infty\le \varepsilon/4 \quad \text{for every $v\in\Delta_L$}.
\]
Using $\cA$ and $\cB$, define
\[
f(x):=-r_\varphi(g(x)).
\]
This function can be evaluated in time $O(T+T_\Phi)$, since $\cA$ computes $g(x)$ in time $T$ and $\cB$ computes $r_\varphi(g(x))$ in time $T_\Phi$.

Now fix $\mu,g^*$, and $s$ satisfying $\langle s-g^*,f\rangle\le \varepsilon/2$. As in the reverse-direction proof of \Cref{thm:intro-main-1} in \Cref{sec:proof-main-non-unif}, the definitions of $\H_\varphi$ and $\D_\varphi$ imply
\[
\D_\varphi(g^*\|g)=\D_\varphi(s\|g)+\langle s-g^*,\nabla\varphi\circ g\rangle+(\H_\varphi(s)-\H_\varphi(g^*)).
\]
Since $\D_\varphi(s\|g)\ge 0$ and $\|s(x)-g^*(x)\|_1\le 2$ for every $x\in\zo^n$, the approximation guarantee $\|r_\varphi-\nabla\varphi\|_\infty\le\varepsilon/4$ implies
\[
\D_\varphi(g^*\|g)\ge \langle s-g^*,r_\varphi\circ g\rangle+(\H_\varphi(s)-\H_\varphi(g^*))-\varepsilon/2.
\]
Finally, by the definition of $f$,
\[
\langle s-g^*,r_\varphi\circ g\rangle=-\langle s-g^*,f\rangle\ge -\varepsilon/2.
\]
Combining the last two inequalities gives
\[
\D_\varphi(g^*\|g)\ge \H_\varphi(s)-\H_\varphi(g^*)-\varepsilon,
\]
which is equivalent to the desired inequality.
\end{proof}

\section{Handling Unbounded Subgradients}
\label{sec:unbounded}
Our \Cref{thm:intro-main-1} requires the assumption that the subgradient $\nabla \varphi$ is bounded. This assumption is not satisfied by the Shannon entropy, where $\varphi(v) = \sum_i v_i\ln v_i$. In this case, $\nabla \varphi(v) = (\ln(v_1),\ldots,\ln(v_L))$ becomes unbounded when some $v_i$ approaches zero. We show that this can be addressed using the following more general theorem than \Cref{thm:intro-main-1}. Here, we consider an arbitrary transformation $\sigma_\varphi:\Delta_L \to \Delta_L$ associated with each $\varphi\in \Phi$, and \Cref{thm:intro-main-1} is the special case when these transformations $\sigma_\varphi$ are the identity function. We will later choose these transformations so that the subgradient $\nabla \varphi(\sigma_\varphi(v))$ at the transformed point $\sigma_\varphi(v)$ is bounded, even when $\nabla \varphi(v)$ may not be bounded.
\begin{theorem}
\label{thm:unbounded}
Let $T > 0$ be a size bound and let $\varepsilon\in (0,1)$ be an error parameter.
Let $\Phi$ be a family of convex functions $\varphi:\Delta_L\to \R$.
For every $\varphi\in \Phi$, let $\sigma_\varphi:\Delta_L\to \Delta_L$ be an arbitrary transformation.
Assume that given $v\in \Delta_L$, the transformed subgradient $\nabla \varphi(\sigma_\varphi(v))$ is bounded in $[-1,1]^L$ and can be computed to $\ell_\infty$ accuracy $\eps/4$ by circuits of size at most $T_\Phi$, for every $\varphi\in \Phi$.
Then there exists
\begin{equation}
T' = O\left(\frac{(T + T_\Phi)\log L}{\varepsilon^2} + L\cdot \poly(1/\varepsilon,\log L)\right)
\end{equation}
such that for every $g^*:\zo^n\to \Delta_L$, there exists $s\in \ind(g^*;T,\varepsilon)\cap \allowbreak \Time(T')$
such that
    \begin{equation}
    \label{eq:unbounded-1}
    \H_\varphi(\sigma_\varphi\circ s) - \H_\varphi(g^*) \ge \D_\varphi(g^*\| \sigma_\varphi\circ s) - \varepsilon \quad \text{for every $\varphi\in \Phi$.}
    \end{equation}
Conversely, for every $T > 0$, there exists $T' = O(T + T_\Phi)$ such that for every $g^*:\zo^n\to \Delta_L$ and every $\varphi\in \Phi$,
\begin{equation}
\label{eq:unbounded-2}
\max_{s\in \ind(g^*;T',\varepsilon/2)}(\H_\varphi(s) - \H_\varphi(g^*)) \le \min_{g\in \Time(T)}\D_\varphi(g^*\|\sigma_\varphi\circ g) + \varepsilon.
\end{equation}
\end{theorem}

\Cref{thm:unbounded} can be proved in the same way as \Cref{thm:intro-main-1} using \Cref{thm:reg-non-unif}, and similarly, we can prove a generalized version of \Cref{thm:unif} using \Cref{thm:reg-unif} to handle unbounded subgradients in the uniform setting.

The advantage of using \Cref{thm:unbounded} is that it only requires $\nabla \varphi(\sigma_\varphi(v))$ to be bounded after the transformation $v\mapsto \sigma_\varphi(v)$. The challenge, however, is that while \eqref{eq:unbounded-1} is stated for the transformed $\sigma_\varphi\circ s$, \Cref{thm:unbounded} only ensures that the untransformed $s$ belongs to $\ind(g^*;T,\varepsilon) \cap \Time(T')$.
This mismatch between $s$ and $\sigma_\varphi\circ s$ also appears in the reverse direction \eqref{eq:unbounded-2} where we minimize the divergence of the transformed $\sigma_\varphi\circ g$ subject to the untransformed $g$ being in $\Time(T)$.
Below we show how to address this mismatch for the Shannon entropy by choosing the transformation $\sigma_\varphi$ appropriately.

For the (negative) Shannon entropy $\varphi$ in \eqref{eq:mw-varphi}, define $\sigma_\varphi(v) = (1 - \varepsilon)v + \varepsilon u$, where $u = (1/L,\ldots,1/L)$ is the uniform distribution and $\varepsilon\in (0,1/2)$. Intuitively, this transformation perturbs $v$ away from the boundary of $\Delta_L$. We now have $\|\nabla \varphi(\sigma_\varphi(v))\|_\infty \le \ln(L/\varepsilon)$, so the boundedness assumption of \Cref{thm:unbounded} is satisfied after a multiplicative scaling of $1/\ln(L/\varepsilon)$.
Moreover, we have the following implication connecting $\sigma_\varphi\circ s$ with $s$ that addresses the mismatch:
\[
s\in \ind(g^*;T,\varepsilon)\cap \Time(T') \Longrightarrow \sigma_\varphi\circ s\in \ind(g^*;T,3 \varepsilon)\cap \Time(O(T'\,\polylog (L/\varepsilon))).
\]
Indeed, it is clear that $\|\sigma_\varphi(v) - v\|_1 \le 2 \varepsilon$, so $s\in \ind(g^*;T,\varepsilon)$ implies that $\sigma_\varphi\circ s\in \ind(g^*;T,3\varepsilon)$. Also, given that $s$ can be computed in circuit complexity $T'$, it is clear that $\sigma_\varphi\circ s$ can be computed in circuit complexity $O(T'\,\polylog (L/\varepsilon))$. We can similarly address the mismatch in the reverse direction \eqref{eq:unbounded-2} using the following connection between $\sigma_\varphi\circ g$ and $g$:
\[
\D_\varphi(g^*\|\sigma_\varphi \circ g) \le \D_\varphi(g^*\|g) + O(\varepsilon),
\]
which holds because for every $v^*,v\in \Delta_L$,
\[
\D_\varphi(v^*\|\sigma_\varphi(v)) = \sum_i v^*_i\ln\frac{v^*_i}{\sigma_\varphi(v)_i} \le \sum_i v^*_i\ln\frac{v^*_i}{(1-\varepsilon)v_i} = \D_\varphi(v^*\|v) + \ln\frac 1{1 - \varepsilon}  \le \D_\varphi(v^*\|v)  + O(\varepsilon).
\]

\section{Exponential Lower Bound for a Particular Entropy Notion}
\label{sec:lb}

We prove \Cref{thm:lb-exp} and \Cref{cor:lb} in this section.
We need the following helper lemma:
\begin{lemma}[\cite{EFF}]
\label{lm:lb-helper}
Define $\alpha = 1/16$. Let $n$ be a sufficiently large positive integer. Then there exist $m = \exp(\Omega(n))$ subsets $S_1,\ldots,S_m \subseteq [n]:= \{1,\ldots,n\}$ such that
\begin{enumerate}
\item $|S_i|\ge \alpha n$ for every $i = 1,\ldots,m$;
\item \label{item:intersection} $|S_i \cap S_j| \le 2\alpha^2 n$ for every distinct $i,j\in \{1,\ldots,m\}$;
\item $|S_i|$ is even for every $i = 1,\ldots,m$.
\end{enumerate}
\end{lemma}
This lemma can be proved by a natural argument using the probabilistic method (see e.g.\ Problem 3.2 of \cite{pseudorandomness}).

\begin{proof}[Proof of \Cref{thm:lb-exp}]

Every $x\in \zo^n = \zo^L$ can be viewed as the indicator of a subset $S_x$ of $[L]$ where for every $y\in [L]$,
\[
y\in S_x\Longleftrightarrow x_y = 1.
\]
We define $\alpha := 1/16$. By \Cref{lm:lb-helper}, there exists a subset $\cX'\subseteq \zo^n$ such that
\begin{enumerate}
    \item $|\cX'| = \exp(\Omega(n))$;
    \item $|S_x| \ge \alpha L$ for every $x\in \cX'$;
    \item $|S_x\cap S_{x'}|\le 2\alpha^2 L$ for every distinct $x,x'\in \cX'$;
    \item $|S_x|$ is even for every $x\in \cX'$.
\end{enumerate}
    For every $x\in \cX'$, we arbitrarily partition $S_x$ into two equal sized subsets $S_x\sps 0$ and $S_x\sps 1$. We use $x\sps 0,x\sps 1\in \zo^n$ to denote the indicators of $S_x\sps 0,S_x\sps 1$, respectively. That is, $S_{x\sps 0} = S_x\sps 0,S_{x\sps 1} = S_x\sps 1$.

    We choose $\mu$ to be the uniform distribution on $\cX'$.
Let $\eta:\cX'\to \zo$ be a function that we will determine later. For every $x\in \cX'$, we define $g^*(x)\in \Delta_L$ to be the uniform distribution on $S_x\sps {\eta(x)}\subseteq [L]$. For every $x\in \zo^n\setminus \cX'$, we define $g^*(x)$ arbitrarily.  As required in the theorem statement, we define $f:\zo^n\to [-1,1]^L$ to be the identity function: $f(x) = x \in \zo^n\subseteq [-1,1]^L$ for every $x\in \zo^n$. We choose the convex function $\varphi:\Delta_L\to \R$ and its corresponding subgradient as follows: for every $v\in \Delta_L$,
\begin{align*}
    \varphi(v) & := \max_{x\in \cX'}\langle v, x\sps{1 - \eta(x)}\rangle,\\
    \nabla \varphi(v) & := x_v\sps {1 - \eta(x_v)}\in \zo^L, \quad \text{where $x_v = \argmax_{x\in \cX'} \langle v, x\sps{1 - \eta(x)}\rangle$}.
\end{align*}
We break ties arbitrarily in the $\argmax$ above.

We first show that
\begin{equation}
\label{eq:lb-3}
-\H_\varphi(g^*)\le 4\alpha.
\end{equation}
By the definition of $\H_\varphi$, we have
\begin{equation}
\label{eq:lb-2}
-\H_\varphi(g^*) = \E_{x'\sim\mu}\left[\max_{x\in \cX'}\langle g^*(x'), x\sps {1 - \eta(x)}\rangle\right].
\end{equation}
Recall that $g^*(x')$ is the uniform distribution on $S_{x'}\sps{\eta(x')}$, so every coordinate of $g^*(x')$ is at most $1/|S_{x'}\sps{\eta(x')}| \le 1/(\alpha L /2)$, and all coordinates out of $S_{x'}\sps{\eta(x')}$ are zero. Therefore,
\begin{equation}
\label{eq:lb-1}
\langle g^*(x'), x\sps {1 - \eta(x)}\rangle \le \frac{1}{\alpha L/2}\cdot |S_{x'}\sps {\eta(x')}\cap S_x\sps {1 - \eta(x)}|.
\end{equation}
If $x = x'$, we have $|S_{x'}\sps {\eta (x')}\cap S_x\sps {1 - \eta(x)}| = 0$. If $x\ne x'$, we have
\[
|S_{x'}\sps {\eta (x')}\cap S_x\sps {1 - \eta (x)}| \le |S_{x'}\cap S_x|\le 2\alpha^2 L.
\]
Plugging this into \eqref{eq:lb-1}, we get
\[
\langle g^*(x'), x\sps {1 - \eta (x)}\rangle \le 4\alpha.
\]
Combining this with \eqref{eq:lb-2}, we get \eqref{eq:lb-3}, as desired.

    Now let $s:\zo^n\to \Delta_L$ be an arbitrary function that is $(\{f\},0.05)$ indistinguishable from $g^*$ and satisfies \eqref{eq:larger-entropy}. By the indistinguishability, we have
    \[
    \E_{x\sim \mu}\langle s(x),x\rangle = \E_{x\sim \mu}\langle s(x),f(x)\rangle \ge \E_{x\sim \mu}\langle g^*(x),f(x)\rangle - 0.05 = 1 - 0.05.
    \]

By \eqref{eq:larger-entropy},
\begin{align*}
\E_{x\sim \mu}\langle s(x),x\sps {1 - \eta(x)}\rangle \le -\H_\varphi(s) \le -\H_\varphi(g^*) + 0.05 \le 4\alpha + 0.05 = 1/4 + 0.05.
\end{align*}
Combining the two inequalities above, we have
\begin{equation}
\label{eq:lb-4}
\E_{x\sim \mu}\langle s(x),x\sps {\eta (x)}\rangle =  \E_{x\sim \mu}\langle s(x),x\rangle - \E_{x\sim \mu}\langle s(x),x\sps {1 - \eta (x)}\rangle \ge 0.6.
\end{equation}

It remains to show that there exists a function $\eta :\cX'\to \zo$ such that every function $s:\zo^n\to \Delta_L$ satisfying \eqref{eq:lb-4} must have circuit complexity $\exp(\Omega(n))$. Let $T$ be a size bound such that for every function $\eta :\cX'\to \zo$, there exists a function $s:\zo^n\to \Delta_L$ satisfying \eqref{eq:lb-4} with circuit complexity at most $T$ (i.e., $s\in \Time (T)$). Our goal is to show that $T = \exp(\Omega(n))$.

We use a probabilistic counting argument. For every fixed function $s:\zo^n\to \Delta_L$ and every fixed $x\in \cX'$, for a uniformly random $\eta :\cX'\to \zo$, we have
\[
\E_\eta \langle s(x), x\sps{\eta (x)}\rangle = \frac 12 (\langle s(x), x\sps{0}\rangle + \langle s(x), x\sps{1}\rangle) = \frac 12 \langle s(x),x\rangle \le 0.5.
\]
Since $\eta (x)$ is distributed independently for different $x\in \cX'$, by the Chernoff bound, for every fixed $s:\zo^n\to \Delta_L$,
\[
\Pr\nolimits_\eta [\E_{x\sim \mu}\langle s(x),x\sps {\eta (x)} \rangle \ge 0.6] \le \exp(-\Omega(|\cX'|)).
\]
There are $\exp(\poly(T))$ circuits of size at most $T$. Therefore, by the union bound,
\[
1 = \Pr\nolimits_\eta [\exists s\in \Time (T) \text{ s.t.\ } \E_{x\sim \mu}\langle s(x),x\sps {\eta (x)} \rangle \ge 0.6] \le \exp(-\Omega(|\cX'|)) \cdot \exp(\poly(T)).
\]
Since $|\cX'| = \exp(\Omega(n))$, the above inequality implies that $T = \exp(\Omega(n))$ for every sufficiently large $n$, as desired.
\end{proof}
\begin{proof}[Proof of \Cref{cor:lb}]
It suffices to show that the two conditions (multiaccuracy and weight-restricted calibration) satisfied by $s$ imply that $s$ belongs to $\ind(g^*;\{f\}, 0.05)$ and satisfies \eqref{eq:larger-entropy} as required by \Cref{thm:lb-exp}.

Indeed, the multiaccuracy condition is equivalent to $s\in \ind(g^*;\{f\},0.05)$.
By \eqref{eq:gap-calibration}, the weight-restricted calibration condition is equivalent to
\[
\H_\varphi(s) - \H_\varphi(g^*) \ge \D_\varphi(g^*\|s) - 0.05,
\]
which, by the non-negativity of the Bregman divergence, implies \eqref{eq:larger-entropy}, as desired.
\end{proof}

\section{Multicalibration Theorems}

\label{sec:mc}
In this section, we show how various forms of the Multicalibration Theorem follow from our Enhanced Regularity/Leakage-Simulation Lemma.  These improve existing formulations of the Multicalibration Theorem in the literature (e.g. \cite{omni}) in the dependence of oracle and circuit complexity on the alphabet size $L$.


\begin{theorem}[Nonuniform Multicalibration Theorem]
\label{thm:mc-non-unif}
    Let $\cF$ be a family of distinguishers $f:\zo^n\to [-1,1]^L$. For every distribution $\mu$ on $\zo^n$, every $g^*:\zo^n\to \Delta_L$, and every $\varepsilon\in (0,1/2)$, there exists a simulator $s:\zo^n\to \Delta_L$ such that
    \begin{enumerate}
        \item $s$ has a $\cF$-oracle circuit of ordinary circuit size
        \[
        O\left(\frac 1\varepsilon\right)^{L + 1} +
        L\cdot \poly(1/\varepsilon,\log L)
        \]
        and uses $O((\log L)/\varepsilon^2)$ $\cF$-oracle gates;
        \item the range of $s$ has size at most $O(1/\varepsilon)^{L-1}$;
        \item $s$ is multicalibrated w.r.t.\ $\cF$:
        \[
        \sum_{v\in \mathsf{range}(s)}|\E_{x\sim \mu}[\langle s(x) - g^*(x),f(x)\rangle \cdot \I[s(x) = v]]| \le \varepsilon\quad \text{for every $f\in \cF$},
        \]
        where the sum is over the finitely many $v\in \Delta_L$ in the range of $s$.
    \end{enumerate}
    In particular, if every $f\in\cF$ has circuit complexity at most $T$, then $s$ has ordinary circuit complexity
    \[
    O\left(\frac{T \log L}{\varepsilon^2} \right) +
    O\left(\frac 1\varepsilon\right)^{L + 1} +
    L\cdot \poly(1/\varepsilon,\log L).
    \]
\end{theorem}
\begin{proof}
Let $\eta=\varepsilon/2$, and let $V\subseteq \Delta_L$ be an $\eta$-covering of $\Delta_L$ in $\ell_1$ distance with $|V|\le O(1/\eta)^{L-1}$.
Fix a map $\sigma:\Delta_L\to V$ such that
\[
\|\sigma(v)-v\|_1\le \eta\quad \text{for every $v\in \Delta_L$}.
\]
For every $f\in\cF$ and every function $r:V\to\{-1,1\}$, define a general distinguisher
\[
\tau_{f,r}(x,v):=r(\sigma(v)) f(x),
\]
and let $\gendistfamily$ be the family of all such distinguishers. Since $V$ has size at most $O(1/\eta)^{L-1}$, the map $\sigma$ and the arbitrary signs $r$ can be implemented by a lookup table over $V$ with ordinary circuit size $O(1/\eta)^{L-1}$, plus one $\cF$-oracle gate for evaluating $f$.

We apply \Cref{thm:reg-non-unif} to the family $\gendistfamily$ with error parameter $\varepsilon/2$. This gives a simulator $s':\zo^n\to \Delta_L$ such that
\begin{equation}
\label{eq:mc-proof-sprime}
\langle s'-g^*,(\tau_{f,r})_{s'}\rangle\le \varepsilon/2
\quad \text{for every $f\in\cF$ and every $r:V\to\{-1,1\}$}.
\end{equation}
Moreover, replacing each $\gendistfamily$-oracle gate by its implementation using one $\cF$-oracle gate and the lookup table above gives a $\cF$-oracle circuit for $s'$ with ordinary circuit size at most
\[
O\left(\frac{\log L}{\varepsilon^2}\right)\cdot O(1/\varepsilon)^{L-1}
+ L\cdot \poly(1/\varepsilon,\log L).
\]
The first term is $O(1/\varepsilon)^{L+1}\log L$, which is absorbed into $O(1/\varepsilon)^{L+1}$ since $\varepsilon<1/2$ and $\log L\le O(2^{L+1})$. Thus we obtain the stated ordinary circuit size bound. The number of $\cF$-oracle gates is $O((\log L)/\varepsilon^2)$.
We output the rounded simulator
\[
s:=\sigma\circ s'.
\]
The additional circuit complexity needed to compute $\sigma$ is absorbed in the bound above.
Since $s$ takes values in $V$, its range has size at most $|V|\le O(1/\eta)^{L-1}=O(1/\varepsilon)^{L-1}$.

It remains to prove multicalibration. Fix any $f\in\cF$. Since $s$ takes values in $V$, it suffices to bound the sum over $v\in V$.
For every $v\in V$, define
\[
A_v:=\E_{x\sim\mu}[\langle s'(x)-g^*(x),f(x)\rangle\cdot \I[s(x)=v]].
\]
Choose $r_f:V\to\{-1,1\}$ so that $r_f(v)=1$ when $A_v\ge 0$ and $r_f(v)=-1$ otherwise.
Then $\tau_{f,r_f}\in\gendistfamily$, and by \eqref{eq:mc-proof-sprime},
\[
\sum_{v\in V}|A_v|
=\langle s'-g^*,(\tau_{f,r_f})_{s'}\rangle
\le \varepsilon/2.
\]
Now define
\[
B_v:=\E_{x\sim\mu}[\langle s(x)-g^*(x),f(x)\rangle\cdot \I[s(x)=v]].
\]
Using $\|f(x)\|_\infty\le 1$ and $\|s(x)-s'(x)\|_1\le \eta$ for every $x$, we have
\begin{align*}
\sum_{v\in V}|B_v|
&\le \sum_{v\in V}|A_v|
+ \sum_{v\in V}\left|\E_{x\sim\mu}[\langle s(x)-s'(x),f(x)\rangle\cdot \I[s(x)=v]]\right|\\
&\le \varepsilon/2+\E_{x\sim\mu}[\|s(x)-s'(x)\|_1\cdot \|f(x)\|_\infty]\\
&\le \varepsilon/2+\eta
=\varepsilon.
\end{align*}
This is exactly the required multicalibration condition.
\end{proof}
In the above formulation of multicalibration, we consider the function $s$ to be partitioning the domain $\zo^n$ into $O(1/\eps)^{L-1}$ nonempty pieces $s^{-1}(v)$.  Then for every distinguisher $f$, we consider its (nonnegative) advantage in distinguishing $s$ from $g^*$ on each piece, and average that advantage over all of the pieces. 
In the next formulation, we allow the distinguisher $f$ to depend adaptively on the choice of piece $s^{-1}(v)$.

\begin{theorem}[Adaptive Nonuniform Multicalibration Theorem]
\label{thm:adaptive-mc-non-unif}
    Let $\cF$ be a family of distinguishers $f:\zo^n\to [-1,1]^L$. For every distribution $\mu$ on $\zo^n$, every $g^*:\zo^n\to \Delta_L$, and every $\varepsilon,\eta\in (0,1/2)$, there exists a simulator $s:\zo^n\to \Delta_L$ with a $\cF$-oracle circuit of ordinary circuit size
    \[
    L\cdot \poly(1/\varepsilon,\log(L/\eta))
    \]
    using $O((\log L)/\varepsilon^2)$ $\cF$-oracle gates,
    and whose range has size at most
    \[
    \binom{\lceil 2L/\eta\rceil+L-1}{L-1}\le O(1/\eta)^{L-1},
    \]
    such that for every $f\in\cF$ and every $v\in \Delta_L$,
    \begin{equation}
    \label{eq:adaptive-mc}
    \left|\E_{x\sim\mu}[\langle s(x)-g^*(x),f(x)\rangle\cdot \I[s(x)=v]]\right|
    \le \varepsilon+\eta\cdot \Pr_{x\sim\mu}[s(x)=v].
    \end{equation}
    Equivalently, whenever $\Pr_{x\sim\mu}[s(x)=v]>0$,
    \[
    \left|\E_{x\sim\mu}[\langle s(x)-g^*(x),f(x)\rangle\mid s(x)=v]\right|
    \le \eta+\frac{\varepsilon}{\Pr_{x\sim\mu}[s(x)=v]}.
    \]
    In particular, if every $f\in\cF$ has circuit complexity at most $T$, then $s$ has ordinary circuit complexity
    \[
    O\left(\frac{T\log L}{\varepsilon^2}
    + L\cdot \poly(1/\varepsilon,\log(L/\eta))\right).
    \]
\end{theorem}
Think of $\eps=\eta^L$.  On average over the $O(1/\eta)^{L-1}$ choices for $v$, we expect $\Pr_{x\sim \mu}[s(x)=v]$ to be $\Omega(\eta)^{L-1}$, the distinguishing advantage on piece $v$ will be
$$\eta+\frac{\varepsilon}{\Pr_{x\sim\mu}[s(x)=v]}
= \frac{\eta^L}{\Omega(\eta)^{L-1}} = O(\eta).$$

\begin{proof}
If $L=1$, the statement is trivial because $\Delta_L$ contains only a single point. Thus, assume $L\ge 2$.
Let $q=\lceil 2L/\eta\rceil$ and let
\[
V:=\{(a_1/q,\ldots,a_L/q)\in \Delta_L\mid a_1,\ldots,a_L\in \Z_{\ge 0}\}.
\]
We define the rounding map $\sigma:\Delta_L\to V$ as follows.
Given $v=(v_1,\ldots,v_L)\in\Delta_L$, let $b_i=\lfloor qv_i\rfloor$ for every $i\in [L]$.
Let
\[
R:=q-\sum_{i=1}^L b_i.
\]
Since $\sum_i qv_i=q$, the number $R$ is an integer between $0$ and $L-1$.
We set $a_i=b_i+1$ for the first $R$ coordinates, and set $a_i=b_i$ for the remaining coordinates.
Then $\sum_i a_i=q$, so $\sigma(v):=(a_1/q,\ldots,a_L/q)$ belongs to $V$.
Moreover, each coordinate changes by at most $1/q$, and hence
\[
\|\sigma(v)-v\|_1\le L/q\le \eta.
\]
The map $\sigma$ can be computed in circuit complexity $L\cdot\polylog(L/\eta)$, since it only uses arithmetic and comparisons on $L$ integers with $O(\log(L/\eta))$ bits.
For every $u\in V$ and $b\in\{-1,1\}$, define $r_{u,b}:V\to\{0,-1,1\}$ by
\[
r_{u,b}(v):=b\cdot \I[v=u].
\]
This function has exactly one nonzero value. When $u$ is hardwired, evaluating it only requires comparing $L$ grid coordinates, so its circuit complexity is $L\cdot\polylog(L/\eta)$, with no dependence on $|V|$.

For every $f\in\cF$, $u\in V$, and $b\in\{-1,1\}$, define
\[
\tau_{f,u,b}(x,v):=r_{u,b}(\sigma(v))f(x),
\]
and let $\gendistfamily$ be the family of all such distinguishers. Each $\tau\in\gendistfamily$ can be implemented using one $\cF$-oracle gate and $L\cdot\polylog(L/\eta)$ ordinary gates. Applying \Cref{thm:reg-non-unif} to $\gendistfamily$ with error parameter $\varepsilon$ gives a simulator $s':\zo^n\to\Delta_L$ such that
\[
\left|\E_{x\sim\mu}[\langle s'(x)-g^*(x),f(x)\rangle\cdot \I[\sigma(s'(x))=u]]\right|\le \varepsilon
\]
for every $f\in\cF$ and every $u\in V$, where we use the two choices $b=1$ and $b=-1$ to obtain the absolute value.
We output
\[
s:=\sigma\circ s'.
\]
Since $s$ takes values in $V$, its range has size at most
\[
|V|=\binom{q+L-1}{L-1}=\binom{\lceil 2L/\eta\rceil+L-1}{L-1}\le O(1/\eta)^{L-1}.
\]
By \Cref{thm:reg-non-unif}, replacing each $\gendistfamily$-oracle gate by the implementation above gives a $\cF$-oracle circuit for $s'$ with ordinary circuit size at most
\[
\frac{L\cdot\polylog(L/\eta)\log L}{\varepsilon^2}
+ L\cdot \poly(1/\varepsilon,\log L).
\]
Since $\varepsilon<1/2$ and $L\ge 2$, the terms above, as well as the additional $L\cdot\polylog(L/\eta)$ cost of evaluating $\sigma$, are absorbed by $L\cdot \poly(1/\varepsilon,\log(L/\eta))$. The number of $\cF$-oracle gates is $O((\log L)/\varepsilon^2)$, giving the stated oracle-circuit bound for $s$.

It remains to prove \eqref{eq:adaptive-mc}. If $v\notin V$, then $\Pr[s(x)=v]=0$ and the left-hand side of \eqref{eq:adaptive-mc} is zero. Now fix $v\in V$. We have
\begin{align*}
&\left|\E_{x\sim\mu}[\langle s(x)-g^*(x),f(x)\rangle\cdot \I[s(x)=v]]\right|\\
&\le
\left|\E_{x\sim\mu}[\langle s'(x)-g^*(x),f(x)\rangle\cdot \I[s(x)=v]]\right|
+\E_{x\sim\mu}[\|s(x)-s'(x)\|_1\cdot \|f(x)\|_\infty\cdot \I[s(x)=v]]\\
&\le \varepsilon+\eta\cdot \Pr_{x\sim\mu}[s(x)=v],
\end{align*}
as required.
\end{proof}

Finally, we now state a version of the Multicalibration Theorem where the simulator is exactly the average value of $g^*$ on each piece of an efficient partition, yielding the functional version of the statement from the introduction (\Cref{thm:intro-mc-rv-informal-ours}).

\begin{theorem}[Average-Preserving Nonuniform Multicalibration Theorem]
\label{thm:intro-mc-rv-formal}
Let $\cF$ be a family of distinguishers $f:\zo^n\to[-1,1]^L$.
For every distribution $\mu$ on $\zo^n$, every $g^*:\zo^n\to\Delta_L$, and every $\varepsilon\in(0,1/2)$, there exists a partition function $P:\zo^n\to[K]$ such that
\[
K\le \left(O(1/\varepsilon)\right)^{L-1},
\]
and $P$ has a $\cF$-oracle circuit of ordinary circuit size
\[
 \left(O(1/\varepsilon)\right)^{L+1}
 + L\cdot\poly(1/\varepsilon,\log L)
\]
and uses $O((\log L)/\varepsilon^2)$ $\cF$-oracle gates.
Moreover, if $\bar s:\zo^n\to\Delta_L$ is defined by
\[
\bar s(x):=\E_{x'\sim\mu}\left[g^*(x')\mid P(x')=P(x)\right],
\]
with arbitrary values on empty pieces of the partition, then
\[
\sum_{a=1}^K
\left|\E_{x\sim\mu}\left[\langle \bar s(x)-g^*(x),f(x)\rangle\cdot\I[P(x)=a]\right]\right|
\le \varepsilon
\quad \text{for every } f\in\cF.
\]
In particular, if every $f\in\cF$ has circuit complexity at most $T$, then $P$ has ordinary circuit complexity
\[
O\left(\frac{T\log L}{\varepsilon^2}\right)
 + \left(O(1/\varepsilon)\right)^{L+1}
 + L\cdot\poly(1/\varepsilon,\log L).
\]
\end{theorem}
\begin{proof}
Let $\alpha=\eta=\varepsilon/4$, and let $V\subseteq\Delta_L$ be an $\eta$-covering of $\Delta_L$ in $\ell_1$ distance with
\[
|V|\le O(1/\eta)^{L-1}\le \left(O(1/\varepsilon)\right)^{L-1}.
\]
Fix a map $\sigma:\Delta_L\to V$ such that
\[
\|\sigma(v)-v\|_1\le \eta\quad \text{for every $v\in\Delta_L$}.
\]
Let $\mathcal{Z}=\{-1,1\}^L$ be the set of vertices of the $\ell_\infty$ unit ball in $\R^L$.
We define a family $\gendistfamily$ of general distinguishers containing the following two types.
First, for every $f\in\cF$ and every function $r:V\to\{-1,1\}$, include
\[
\tau_{f,r}(x,v):=r(\sigma(v))f(x).
\]
Second, for every function $h:V\to\mathcal{Z}$, include
\[
\tau_h(x,v):=h(\sigma(v)).
\]
The second type is the strengthened average-preservation test: after the simulator is rounded to $V$, it allows a separate vertex of the $\ell_\infty$ unit ball to be chosen on each rounded piece.
Since $V$ has size at most $\left(O(1/\varepsilon)\right)^{L-1}$, the map $\sigma$, the scalar signs $r$, and the vertex-valued maps $h$ can be implemented by lookup tables using ordinary circuit size at most
\[
L\cdot \left(O(1/\varepsilon)\right)^{L-1},
\]
plus one $\cF$-oracle gate for distinguishers of the first type and no $\cF$-oracle gates for distinguishers of the second type.

Apply \Cref{thm:reg-non-unif} to $\gendistfamily$ with error parameter $\alpha$ to obtain a simulator $s':\zo^n\to\Delta_L$ satisfying
\begin{equation}
\label{eq:avg-mc-general-dist}
\langle s'-g^*,\tau_{s'}\rangle\le \alpha
\quad \text{for every $\tau\in\gendistfamily$}.
\end{equation}
We output the rounded simulator
\[
s:=\sigma\circ s'.
\]
Let $P$ be an arbitrary indexing of the nonempty level sets of $s$ by elements of $[K]$.
Then $K\le |V|\le \left(O(1/\varepsilon)\right)^{L-1}$.
Computing $P$ only requires computing $s$ and then indexing its value in the finite range of $s$.
By \Cref{thm:reg-non-unif}, the $\gendistfamily$-oracle circuit for $s'$ uses $O((\log L)/\alpha^2)=O((\log L)/\varepsilon^2)$ $\gendistfamily$-oracle gates. Replacing each such gate by the implementation above gives ordinary circuit size at most
\[
O\left(\frac{\log L}{\varepsilon^2}\right)\cdot L\cdot \left(O(1/\varepsilon)\right)^{L-1}
+ L\cdot \poly(1/\varepsilon,\log L).
\]
The first term is $L\log L\cdot\left(O(1/\varepsilon)\right)^{L+1}$, which is absorbed into $\left(O(1/\varepsilon)\right)^{L+1}$ since $\varepsilon<1/2$ and $L\log L\le O(2^{L+1})$. Thus we obtain the stated ordinary circuit size bound. The number of $\cF$-oracle gates is at most the number of $\gendistfamily$-oracle gates, namely $O((\log L)/\varepsilon^2)$.

For every $a\in[K]$ with $\Pr[P(x)=a]>0$, let $s_a$ denote the value of $s(x)$ on the piece $P^{-1}(a)$ and let
\[
\bar s_a:=\E_{x\sim\mu}\left[g^*(x)\mid P(x)=a\right].
\]
For an empty piece, choose $\bar s_a$ arbitrarily in $\Delta_L$.
Then $\bar s(x)=\bar s_{P(x)}$.
We first prove multicalibration for $s$ with respect to the original family $\cF$. Fix $f\in\cF$ and, for every $a\in[K]$, define
\[
A_a:=\E_{x\sim\mu}\left[\langle s'(x)-g^*(x),f(x)\rangle\cdot\I[P(x)=a]\right].
\]
Choose $r_f:V\to\{-1,1\}$ so that $r_f(s_a)=1$ if $A_a\ge 0$ and $r_f(s_a)=-1$ otherwise, with arbitrary values off the range of $s$. Applying \eqref{eq:avg-mc-general-dist} to $\tau_{f,r_f}$ gives
\[
\sum_{a=1}^K |A_a|\le \alpha.
\]
Therefore, using $\|s(x)-s'(x)\|_1\le \eta$ and $\|f(x)\|_\infty\le 1$,
\begin{equation}
\label{eq:avg-mc-s}
\sum_{a=1}^K
\left|\E_{x\sim\mu}\left[\langle s(x)-g^*(x),f(x)\rangle\cdot\I[P(x)=a]\right]\right|
\le \alpha+\eta.
\end{equation}

Next we bound the average distance between $s$ and $\bar s$. For every nonempty piece $a$, choose $z_a\in\mathcal{Z}$ such that
\[
\langle s_a-\bar s_a,z_a\rangle=\|s_a-\bar s_a\|_1,
\]
and define $h:V\to\mathcal{Z}$ by setting $h(s_a)=z_a$ on the range of $s$, with arbitrary values off the range. Applying \eqref{eq:avg-mc-general-dist} to $\tau_h$ and again using $\|s(x)-s'(x)\|_1\le\eta$ gives
\begin{equation}
\label{eq:avg-mc-distance}
\E_{x\sim\mu}\left[\|s(x)-\bar s(x)\|_1\right]
=\E_{x\sim\mu}\left[\langle s(x)-g^*(x),h(s(x))\rangle\right]
\le \alpha+\eta.
\end{equation}

Now fix any $f\in\cF$. Using \eqref{eq:avg-mc-s} and \eqref{eq:avg-mc-distance}, we get
\begin{align*}
\sum_{a=1}^K
\left|\E_{x\sim\mu}\left[\langle \bar s(x)-g^*(x),f(x)\rangle\cdot\I[P(x)=a]\right]\right|
&\le
\sum_{a=1}^K
\left|\E_{x\sim\mu}\left[\langle s(x)-g^*(x),f(x)\rangle\cdot\I[P(x)=a]\right]\right|\\
&\qquad
+\E_{x\sim\mu}\left[\|s(x)-\bar s(x)\|_1\cdot\|f(x)\|_\infty\right]\\
&\le 2(\alpha+\eta)\\
&=\varepsilon,
\end{align*}
as desired.
\end{proof}


\bibliographystyle{alphaurl}
\bibliography{ref}

\appendix
\section{Deferred Proofs}
\subsection[Proof of Lemma~\ref{lm:mw-smooth}]{Proof of \Cref{lm:mw-smooth}}
\label{sec:proof-mw-smooth}
We need the following standard definition of strong convexity and two helper claims (\Cref{claim:Gamma,lm:strong-cvx-smooth}).

\begin{definition}[Strong Convexity]
Let $\cG\subseteq \R^L$ be a bounded convex set.
For $\lambda \ge 0$ and a norm $\|\cdot \|$ on $\R^L$, we say a convex function $\varphi:\cG \to \R$ is $\lambda$-\emph{strongly convex} w.r.t.\ $\|\cdot \|$ if for every $g_0,g\in \cG$, letting $h_0\in \R^L$ be an arbitrary subgradient of $\varphi$ at $g_0$, we have
\begin{equation}
\label{eq:strong-cvx}
\varphi(g) \ge \varphi(g_0) + \langle g - g_0, h_0\rangle + \frac \lambda 2 \|g - g_0\|^2.
\end{equation}
\end{definition}

\begin{claim}
\label{claim:Gamma}
    Let $\cG\subseteq \R^L$ be a bounded convex set. Let $\varphi:\cG\to \R$ be a closed convex function, and let $\psi:\R^L\to \R$ be its convex conjugate. Suppose $(g_0,h_0)\in \cG\times \R^L$ is a conjugate pair (i.e.\ $\Gamma_{\varphi,\psi}(g_0,h_0) = 0$). Then for every $g\in \cG$ and $h\in \R^L$,
    \begin{align}
    \Gamma(g,h_0) & = \varphi(g) - \varphi(g_0) - \langle g - g_0, h_0\rangle,\label{eq:Gamma-1}\\
    \Gamma(g_0,h) & = \psi(h) - \psi(h_0) - \langle g_0, h - h_0\rangle .\label{eq:Gamma-2}
    \end{align}
\end{claim}
\begin{proof}
    Since $(g_0,h_0)$ is a conjugate pair, we have
        \begin{equation}
    0 = \Gamma_{\varphi,\psi}(g_0,h_0) = \varphi(g_0) + \psi(h_0) - \langle g_0,h_0\rangle.\label{eq:Gamma-5}
    \end{equation}
    Also,
    \begin{align}
        \Gamma_{\varphi,\psi}(g,h_0) &= \varphi(g) + \psi(h_0) - \langle g,h_0\rangle,\label{eq:Gamma-3}\\
        \Gamma_{\varphi,\psi}(g_0,h) &= \varphi(g_0) + \psi(h) - \langle g_0,h\rangle.\label{eq:Gamma-4}
    \end{align}
Taking the difference between \eqref{eq:Gamma-3} and \eqref{eq:Gamma-5} gives \eqref{eq:Gamma-1}. Taking the difference between \eqref{eq:Gamma-4} and \eqref{eq:Gamma-5} gives \eqref{eq:Gamma-2}.
\end{proof}

\begin{claim}
\label{lm:strong-cvx-smooth}
Let $\cG\subseteq \R^L$ be a bounded convex set. For $\lambda > 0$, if $\varphi:\cG\to \R$ is a closed $\frac 1{\lambda}$-strongly convex function w.r.t a norm $\|\cdot \|$ on $\R^L$, then its convex conjugate $\psi$ is $\lambda$-smooth w.r.t.\ the dual norm $\|\cdot \|_*$.
\end{claim}
\begin{proof}
    Consider an arbitrary $h_0\in \R^L$ and let $g_0\in \cG$ be a subgradient of $\psi$ at $h_0$. Since $\varphi$ is closed, we know that $(g_0,h_0)$ is a conjugate pair. By \eqref{eq:Gamma-1}, the strong convexity of $\varphi$ implies that
    \begin{align*}
    \Gamma(g,h_0) & \ge \frac 1{2\lambda}\|g - g_0\|^2 && \text{for every $g\in \cG$},\\
    \text{or equivalently,} \quad
    \Gamma(g + g_0,h_0) & \ge \frac 1{2\lambda}\|g\|^2 && \text{for every $g$ satisfying $g + g_0\in \cG$}.
    \end{align*}
It is easy to verify that the convex conjugate of $\Gamma(g + g_0,h_0)$ (as a function of $g$) is $\Gamma(g_0,h_0 + h)$ (as a function of $h$). It is also straightforward to see that the convex conjugate of $\frac 1{2\lambda}\|g\|^2$ is $\frac \lambda 2\|h\|_*^2$. Therefore, taking the convex conjugate of  both sides of the inequality above, we have
\begin{align*}
\Gamma(g_0,h_0 + h) & \le \frac \lambda{2}\|h\|_*^2 && \text{for every $h\in \R^L$},\\
\hspace{5em}\text{or equivalently,} \quad
\Gamma(g_0,h) & \le \frac \lambda{2}\|h - h_0\|_*^2 && \text{for every $h\in \R^L$}.\hspace{5em}
\end{align*}
By \eqref{eq:Gamma-2}, this proves that $\psi$ is $\lambda$-smooth w.r.t.\ $\|\cdot\|_*$.
\end{proof}
\begin{proof}[Proof of \Cref{lm:mw-smooth}]
Consider the negative Shannon entropy $\varphi$ in \eqref{eq:mw-varphi}.
For $g,g_0\in \Delta_L$, let $h_0:= \nabla \varphi(g_0)$. By \eqref{eq:Gamma-1} and Pinsker's inequality, we have
\[
\varphi(g) - \varphi(g_0) - \langle g - g_0,h_0\rangle = \Gamma_{\varphi,\psi}(g,h_0) = \D_\varphi(g \| g_0) \ge \frac 12 \|g - g_0\|^2_1.
\]
This shows that $\varphi$ is $1$-strongly convex w.r.t.\ $\|\cdot\|_1$.
By \Cref{lm:strong-cvx-smooth}, we know that $\psi$ is $1$-smooth w.r.t.\ $\|\cdot \|_\infty$. Let $(g,h),(\hat g,\hat h)\in \cG\times \R^L$ be two arbitrary conjugate pairs. By \eqref{eq:Gamma-1} and the $1$-strong convexity of $\varphi$ w.r.t.\ $\|\cdot\|_1$, we have
    \[
    \Gamma(g,\hat h) = \varphi(g) - \varphi(\hat g) - \langle g - \hat g,\hat h\rangle \ge \frac 12 \|g - \hat g\|_1^2.
    \]
    Similarly, by \eqref{eq:Gamma-2} and the $1$-smoothness of $\psi$ w.r.t.\ $\|\cdot\|_\infty$, we have
    \[
    \Gamma(g,\hat h) \le \frac 12 \|h - \hat h\|_\infty^2.
    \]
    Combining the two inequalities above, we have
    \[
    \|g - \hat g\|_1 \le \|h - \hat h \|_\infty.
    \]
Since this holds for arbitrary $h,\hat h\in \R^L$ and the gradients $g:= \nabla \psi(h) = \softmax(h), \hat g:= \nabla \psi(\hat h) = \softmax(\hat h)$, we have proved the $1$-Lipschitzness of softmax.
\end{proof}

\end{document}